\documentclass[11pt]{article}
\usepackage[margin=1in,letterpaper]{geometry}
\usepackage{amsmath,amsthm,amssymb}
\usepackage{enumerate}
\usepackage{subcaption}

\usepackage{color}
\usepackage{comment}
\usepackage{tikz}

\usepackage[bookmarks=false,colorlinks=true, linkcolor=blue, urlcolor=blue, citecolor=blue, breaklinks=true]{hyperref}

\newtheorem{theorem}{Theorem}[section]
\newtheorem{lemma}{Lemma}[section]
\newtheorem{corollary}{Corollary}[section]

\newtheorem{proposition}{Proposition}[section]

\newtheorem{definition}{Definition}[section]

\usepackage{todonotes}
\usepackage[linesnumbered]{algorithm2e}
\usepackage{relsize}

\newcommand{\N}{\mathbb{N}}

\newcommand{\BSTs}{\mathcal{T}}
\newcommand{\X}{\mathcal{X}}

\DeclareMathOperator*{\argmax}{arg\,max}

\newcommand{\T}{T}
\newcommand{\floor}[1]{\left\lfloor #1 \right\rfloor}
\newcommand{\ceil}[1]{\left\lceil #1 \right\rceil}

\title{Randomized Binary and Tree Search under Pressure}

\author{Agustín Caracci\thanks{Pontificia Universidad Católica de Chile, Institute for Mathematical and Computational Engineering, Chile. Email: \href{mailto:juan.caracci@uc.cl}{juan.caracci@uc}. Supported by Fondecyt-ANID Nr. 1221460, ANID.} \and
Christoph Dürr\thanks{Sorbonne University, CNRS, LIP6, France. \url{https://www.lip6.fr/Christoph.Durr}. Email: \href{mailto:Christoph.Durr@lip6.fr}{Christoph.Durr@lip6.fr}. Work partially conducted while C.D.\ was affiliated with Universidad de Chile, CMM. Partially supported by the Center for Mathematical Modeling grant FB210005 and the grants ANR-19-CE48-0016, ANR-23-CE48-0010 from the French National Research Agency (ANR).} \and
José Verschae\thanks{Pontificia Universidad Católica de Chile, Institute for Mathematical and Computational Engineering, Chile.  \url{https://sites.google.com/site/jverschae/}. Email: \href{mailto:jverschae@uc.cl}{jverschae@uc.cl}. Partially supported by Fondecyt-ANID Nr. 1221460,
ANID and  by the Center for Mathematical Modeling grant FB210005, Basal, ANID.}}
\begin{document}

\maketitle

\begin{abstract}We study a generalized binary search problem on the line and general trees. On the line (e.g., a sorted array), binary search finds a target node in $O(\log n)$ queries in the worst case, where $n$ is the number of nodes. In situations with limited budget or time, we might only be able to perform a few queries, possibly sub-logarithmic many. In this case, it is impossible to guarantee that the target will be found regardless of its position. Our main result is the construction of a randomized strategy that maximizes the minimum (over the target position) probability of finding the target. Such a strategy provides a natural solution where there is no apriori (stochastic) information of the target's position. As with regular binary search, we can find and run the strategy in $O(\log n)$ time (and using only $O(\log n)$ random bits). Our construction is obtained by reinterpreting the problem as a two-player (\textit{seeker} and \textit{hider}) zero-sum game and exploiting an underlying number theoretical structure.

Furthermore, we generalize the setting to study a search game on trees. In this case, a query returns the edge's endpoint closest to the target. Again, when the number of queries is bounded by some given $k$, we quantify a \emph{the-less-queries-the-better} approach by defining a seeker's profit $p$ depending on the number of queries needed to locate the hider. For the linear programming formulation of the corresponding zero-sum game, we show that computing the best response for the hider (i.e., the separation problem of the underlying dual LP) can be done in time $O(n^2 2^{2k})$, where $n$ is the size of the tree. This result allows to compute a Nash equilibrium in polynomial time whenever $k=O(\log n)$. In contrast, computing the best response for the hider is NP-hard.

\end{abstract}

\section{Introduction}
\label{sec:Introduction}

Searching for an object in an ordered list with $n$ elements is one of the most fundamental tasks in computer science. Binary search, one of the most basic discrete algorithms, yields a search strategy that finds the target item with $O(\log(n))$ comparisons, regardless of its position. However, in situations when comparing item values is expensive, we are faced with the problem of searching under a potentially sub-logarithmic  budget of $k$ comparisons. For small values of $k$, there will be positions of the target value where the strategy will fail. By using a randomized strategy, we can aim to find the target with certain minimum probability, independently of its position. More precisely, we study the problem of devising a randomized strategy that maximizes the worst-case (with respect to the position of the target) probability of finding the target.

The described setting can be equivalently posed as a two-player zero-sum game played on a line graph $G=(V,E)$ with vertex set $V=\{0,1,\ldots,n-1\}$ and edge set $E=\{\{v,v+1\}:0\le v\le n-2\}$. The first player, the \emph{hider} chooses a node $v^*\in V$ to hide (i.e., the target). The second player, the \emph{seeker} must decide on a search strategy, i.e., an adaptive sequence of edges to query that aims to find the hidden node $v^*$. If edge $\{v,v+1\}$ is queried, the seeker learns whether $v^*\le v$. The game finishes after $k$ queries or after the seeker can univocally find the position of $v^*$. If the seeker finds $v^*$ within the given number of queries she gets a profit of 1 and the seeker obtains a profit of -1 (i.e., she pays a cost of 1). If after the $k$ queries the seeker cannot certify the position of $v^*$, then both player get a profit of 0. We are interested in finding a mixed (i.e. randomized) strategy for each player that yields a Nash equilibrium. As this is a zero-sum game, in a Nash equilibrium the seeker chooses a randomized strategy that maximizes $\min_{v^*\in V} P(v^*)$, where $P(v^*)$ denotes the probability that the seeker finds the target given that the hider chooses $v^*$ to hide.

This fundamental problem is motivated by several questions related to the search of substances in water networks. For example, a new trend of water surveillance became prominent during the COVID pandemic, where PCR tests were taken within a sewage water network (normally a tree) in order to efficiently detect virus outbreaks~\cite{ahmed_first_2020,PNAS2024,domokos_identification_2022,larson_sampling_2020,nourinejad_placing_2021,doi:10.5864/d2021-015}. 
Similar strategies have been proposed to detect the presence of illicit drugs~\cite{Sulej-SuchomskaEtAl2020}, pollution~\cite{cristo_pollution_2008} and even illicit explosives~\cite{EMPHASIS}. Observe that if the water network is a path and our aim is to find a single infected node, taking a PCR test on a given edge (or node) will tell us whether the infected node is upstream or downstream. Crucially, in all these scenarios, the number of comparisons is constrained by the time that the sought substance stays in the network. Arguably, these settings induce models where the number of queries is sub-logarithmic, or even constant, with respect to the number of nodes of the network. Hence, maximizing the worst-case probability of finding the infected node is particularly meaningful, especially if there is no known probabilistic data on the source of the target. On the other hand, the (equivalent) game theoretical perspective gives a relevant model where there is an adversary that chooses the location of the target deliberately.

Inspired by these scenarios, we further extend our model in two directions. First, we consider the case in which the graph $G=(V,E)$ is not only a path, but a general tree. As before, the seeker queries edges. If $e\in E$ is queried, she obtains as an answer the connected component of $G-e:=(V,E\setminus\{e\})$ that contains $v^*$. 
As a second extension, we further consider a time-dependent non-increasing profit function $p:\mathbb{Z}_{\ge 0}\rightarrow \mathbb{Z}_{\ge 0}$, where $p(t)=0$ of $t>k$. If the seeker finds the hider within $t$ queries, she obtains a profit of $p(t)$ and the hider perceives a cost of $p(t)$ (or, equivalently, a profit of $-p(t)$). Observe that we obtain the same problem as before if
$p(t)= 1$ for $t\le k$, which we call the \emph{unit profit} case. As before, the objective is to efficiently find a mixed Nash equilibrium.

\paragraph*{Related Literature} 

Previous literature mostly focused on minimizing the worst case number queries for finding the target. When the target is hidden in a tree, Lam and Yue \cite{lam_optimal_2001} (see also~\cite{dereniowski2008edge}) provided a linear time algorithm which generates an optimal search tree for the edge query model, a result that was later rediscovered in a sequence of papers, including the works by Ben-Asher et al.~\cite{ben1999optimal}, Onak and Parys~\cite{onak2006generalization}, and Mozes et al.~\cite{mozes2008finding}. A more general setting was later considered by Cicalese et al.~\cite{cicalese_binary_2012} where each edge has a different cost. The objective is to find the target with the worst case minimum total cost. This version turns out to be strongly NP-hard and admits an approximation algorithm with a slightly sublogarithmic guarantee. 

An alternative model considers queries on the vertex of the tree. If the target is not in the queried vertex, one learns which of the neighbors is closest to the target. For this model, Schäffer provided a linear time algorithm generating an optimal search tree~\cite{schaffer1989optimal}. Minimizing the number of vertex queries on an arbitrary undirected graph becomes quite intriguing. There is an $O(m^{\log n} \cdot n^2 \log n)$ time algorithm, where $m$ is the number of edges and $n$ the number of vertices. The quasi-polynomial part is necessary, as there is no $O(m^{(1-\epsilon)\log n})$ time algorithm under the Strong Exponential-Time Hypothesis (SETH) \cite{emamjomeh-zadeh_deterministic_2016}. For directed acyclic graphs, a vertex query returns 3 types of answers: the target is on the queried vertex, the target is reachable from the queried vertex, or it is not. In this model, minimizing the number of queries is NP-complete \cite{carmo2004searching}.

A different line of research considers a known probability distribution of the position of the target, and aims to find a search strategy (either for node or edge queries) that minimizes the expected number of queries until the target is found.  Knuth~\cite{knuthOptimumBinarySearch1971} showed that dynamic programming can solve this problem on the line with running time $O(n^2)$. For the more complicated problem on trees with the edge query model, Cicalese et al.~\cite{cicalese_complexity_2011,cicalese_improved_2014} showed that finding an optimal search strategy for a tree is NP-hard. Moreover, it admits a 1.62-approximation algorithm and even an FPTAS on tree of bounded degrees. 
On the other hand, for node queries, the complexity of the problem is still open. Recently, Berendsohn and Kozma~\cite{berendsohnSplayTreesTrees2022} showed that the problem admits a PTAS, and Berendsohn et al.~\cite{berendsohnFastApproximationSearch2023} analyzed the algorithm that iteratively select the centroid of the tree. Besides given a fast output-sensitive algorithm for finding the centroid, they show that the obtained strategy is a 2-approximation and that the analysis is tight even when the target distribution is uniform.

\paragraph*{Our Contribution}

We study two scenarios: \emph{(i)} For the unit profit case in the line, we provide an algorithm that computes the optimal strategy for the seeker in $O(\log n)$ time; \emph{(ii)} if $G$ is a general tree and the profit $p$ is arbitrary, we show that we can find a Nash equilibrium in time $2^k\text{poly(n)}$ while computing the best-response for the seeker is NP-hard. 

Our main result is \emph{(i)}. Our algorithm samples a search strategy in time $O(\log n)$ (using $O(\log n)$ random bits). Afterwards, choosing the next node to query takes constant time in each time step. The solution is based on picking a random interval, on which we will perform binary search. The interval is chosen uniformly at random from a set constructed with a greedy algorithm based on modular arithmetics (see Figure~\ref{fig:not-coprime}). This construction naturally connects to Bezout's identity and the Extended Euclidean algorithm, which we exploit to obtain the claimed running times. 
On the other hand, we provide a construction that yields an optimal strategy for the hider, which we then use to show that both strategies (hider and seeker) are optimal by linear programming duality. For some regimes of values of $n$ and $k$, the optimal hider's solution requires an intricate construction as she needs to unevenly distribute probability mass to avoid giving the seeker an advantage (see Figure~\ref{fig:coprime}).    

For \emph{(ii)}, where $G$ is an arbitrary tree and $p$ is any non-increasing profit function, we again consider the corresponding zero-sum game as a linear program (LP), which has an exponential number of variables. We show that the separation of the dual problem --- the problem of computing the best-response of the hider --- can be solved in time $O(n^2 2^{2k})$. In contrast, when $k$ is part of the input the separation problem is NP-hard, even if the diameter or the degree of the tree is constant (see Theorem~\ref{thm:bestResponseHardness}), which justifies the exponential dependency on $k$\footnote{Observe, however, that this does not imply that computing a Nash equilibrium is NP-hard.}. The separation problem of the dual takes as input a distribution of the target node over $V$. It consists of finding a search strategy with at most $k$ queries that maximizes the expected profit of the seeker. We devise a dynamic program that solves this problem. It is based on the concept of \emph{visibility vectors}~\cite{onak2006generalization,lam_optimal_2001, dereniowski2008edge}, introduced for the problem of minimizing the time to find the target in the worst-case. By rooting the tree, we can apply a bottom-up approach to extend a solution for each possible visibility vector. The exponential dependency on $k$ comes from the fact that the number of different visibility vectors is exponential in $k$. Together with the Ellipsoid method, we obtain an algorithm to compute a Nash equilibrium in time $2^k\text{poly}(n)$. 

Observe that in regimes where the budget is logarithmic, $k\le C\cdot \log_2(n)$, then our algorithm for the separation problem runs in polynomial time $O(n^{2+2C})$. Moreover, $k$ is logarithmic in many natural cases. For example, if $G$ has maximum degree $d$ and we are in the unit profit case, then there are search strategies that can unequivocally find the target within $ O(d\log n)$ queries~\cite{ben-asher_cost_1997,emamjomeh-zadeh_deterministic_2016}. Therefore, if $d$ is a constant, our model only makes sense when $k$ is at most logarithmic, as otherwise we have a strategy that finds the hider with probability 1. Hence, for these cases we can assume that $k\le C\cdot \log_2(n)$ and hence our algorithm takes polynomial time $O(n^{2+2C})$. In particular, this yields a running time of $O(n^4)$ for the separation problem when $G$ is a line; which is significantly worse than our dedicated solution.

In Section~\ref{sec:ProbDef} we introduce the problem formally. Our main results are given in Section~\ref{sec:line}, where we show our solution for the case that $G$ is a path. Finally, Section~\ref{sec:DP} shows the dynamic program for the general case. Some of the technical proofs are deferred to the appendix.

\section{Problem Definition }
\label{sec:ProbDef}

For an arbitrary graph $H$, we denote $V(H)$ and $E(H)$ its node and edge set, respectively. We will represent edges as sets $\{u,v\}\in E$, and we will use the shortcut $\{u,v\}=uv$ if appropriate.

Let $G=(V,E)$ be a tree with $n$ nodes. To represent a search strategy, consider a rooted binary out-tree~$\T=(N,D)$, where each internal node $\nu \in N$ is labeled with an edge $e(\nu) \in E$. By definition, the root $\rho\in N$ is labeled, unless $|N|=1$. The edge $e(\rho)=uv$  corresponds to the first edge queried by the strategy. Let $G_u$ (resp. $G_v$) denote the connected component of $G - e := (V,E\setminus \{e\})$ that contains $u$ (resp. $v$). One child of $\rho$ is associated to $G_u$, and the other child to $G_v$. 
The output of the query points to child of $\rho$ associated to the connected component that contains the target. The search is then resumed in said child.
Based on this notation, we give a recursive definition of a search strategy. We say that $\T$ with root label $e(\rho)=uv\in E$ is a search strategy for $G=(V,E)$ if the subtrees of $\rho$ in $\T$, denoted by $T_u$ and $T_v$, are also search strategies for $G_u$ and $G_v$, respectively.

Finally, $\T=(N,D)$ is a search strategy for any tree $G$ if $|N|=1$.

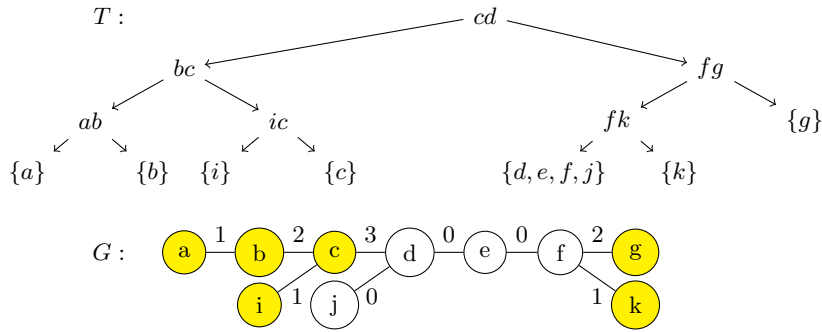
\begin{figure}[htb]
    \centering
\begin{tikzpicture}[yscale=0.7,font=\footnotesize,level distance=1cm,level/.style={sibling distance=50mm/#1}]
\draw (-1, 6) node {$T:$} ;
\draw (4,6) node {$cd$}
child[->, sibling distance=80mm] {node {$bc$}
    child {node {$ab$}
        child {node {$\{a\}$}}
        child {node {$\{b\}$}}
    }
    child {node {$ic$}
        child {node {$\{i\}$}}
        child {node {$\{c\}$}}
        }
}
child[->,sibling distance=60mm] {node {$fg$}
  child {node {$fk$}
             child {node {$\{d,e,f,j\}$}}
             child {node {$\{k\}$}}
        }
  child {node {$\{g\}$}}
};

\begin{scope}[yshift=0.5cm]

\draw (-1,1) node{$G:$}; 
\foreach \v/\x/\y/\c in {a/0/1/yellow, b/1/1/yellow, c/2/1/yellow, d/3/1/white, 
                      e/4/1/white, f/5/1/white, g/6/1/yellow, i/1/0/yellow, 
                      j/2/0/white, k/6/0/yellow} {
    \node[circle,draw,fill=\c] (\v) at (\x,\y) {\v};
};
\foreach \u/\v/\l/\p in {a/b/1/above, b/c/2/above, c/d/3/above, i/c/1/below, j/d/0/below, d/e/0/above, e/f/0/above, f/g/2/above, f/k/1/below} {
    \draw (\u) -- node[\p] {\l} (\v) ;
};
\end{scope}
\end{tikzpicture}    \caption{Example of a search tree $T$ over a tree $G$. Leafs $\nu$ of $T$ are labeled with $V(\nu)$. Covered vertices in $G$ are marked in yellow. Edge labels of $G$ will be explained in Section~\ref{sec:edge-labelings}}
    \label{fig:def-search-tree}
\end{figure}

Let $\T$ be a search strategy for $G=(V,E)$. Observe that each node $\nu$ of $\T$ can be naturally associated to a set $V(\nu)\subseteq V$ that potentially contains the target node $v^*$. More precisely, let us define $V(\rho)$ as $V$. If $e(\rho)=uv$, we define recursively $V(\rho_u)= V(G_u)$ and $V(\rho_v)= V(G_v)$, where $\rho_u$ (resp. $\rho_v$) is the root of $T_u$ (resp. $T_v$) and $V(G_u)$ (resp. $V(G_v)$) is the vertex-set of $G_u$ (resp $G_v$). Hence, at the moment that we query according to node $\nu\in N$, the set $V(\nu)$ corresponds to the smallest node set in $G$ for which we can guarantee that $v^*\in V(\nu)$. 

Let $L\subseteq N$ be the set of leaves of $\T$. It is easy to see that $\{V(\lambda): \lambda \in L\}$ defines a partition into connected components of $V$. We say that node $v \in V$ is \emph{covered} by $\T$ if the singleton $\{v\}$ belongs to $\{V(\lambda): \lambda \in L\}$. The set of nodes covered by $\T$ is denoted by $C(\T)\subseteq V$. Observe that if the hider chooses $v^*$, then applying the search strategy $\T$ we can assert that $v^*$ is the target if and only if $v^*\in C(T)$. 

Finally, the \textit{height} of a search strategy $\T$ is the height of the underlying binary tree, i.e., the length of the longest path from the root $\rho$ to a leaf. For a fixed tree $G$ and budget $k\in \mathbb{N}$, we denote by $\mathcal{T}_k$ the set of all the search strategies for $G$ of height at most $k$. 

Let $\Delta_k$ be the set of distributions supported on $\mathcal{T}_k$, that is,  $\Delta_k:=\{x:  x\ge 0 \text{ and } \sum_{\T\in \mathcal{T}_k}x_{\T}=1 \}$. Our aim is to find a vector $x\in \Delta_k$ that maximizes the worst-case probability of finding the target $v^*\in V$, that is,
\begin{equation}
\label{eq:unitValue}
\max_{x\in \Delta_k} \min_{v^* \in V} \sum_{\T \in \mathcal{T}_k: v^* \in C(\T)} x_{\T}.
\end{equation}
More generally, let $p: \mathbb{Z}_{\ge 0}\rightarrow \mathbb{Z}_{\ge 0}$ be a non-increasing function, where $p(t)$ represents the profit of finding the target with exactly $t$ queries, and $p(t)=0$ for $t>k$.  Our most general model considers the maximization of the worst-case expected profit,
\begin{equation}
\label{eq:arbitraryValue}
\max_{x\in \Delta_k} \min_{v^* \in V} \sum_{\T \in \mathcal{T}_k} x_{\T}\cdot p(h_T(v^*)),
\end{equation}
where $h_{\T}(v^*)$ is $k+1$ if $v^*\not\in C(\T)$ (and hence $p(h_{\T}(v^*))=0$) and, otherwise, it equals the distance from the root $\rho$ to the leaf $\lambda$ such that $V(\lambda)=\{v^*\}$. 

We can interpret this problem as a two-player zero-sum game as follows. The first player, the \textit{seeker}, selects a search strategy $\T\in \mathcal{T}_k$. The second player, the \textit{hider}, selects a node $v\in V$. Hence, for a given pure strategy pair $(\T,v)$, the seeker obtains a profit of $p(h_{\T}(v))$ while the hider incurs a cost of $-p(h_{\T}(v))$ (or a negative profit of $p(h_{\T}(v))$. A mixed-strategy for the seeker is a probability vector $x\in \Delta_k$, where $x_{\T}$ represents the probability of selecting the search strategy $\T$. A mixed-strategy for the hider is a vector $y\in \Delta_V := \{y: y\ge 0 \text{ and } \sum_{v\in V} y_v =1\}$. For a pair $(x,y)$, the expected profit of the seeker (cost of the hider) is 
$ 
p(x,y) = \sum_{\T \in \mathcal{T}_k} \sum_{v\in V} y_v\cdot x_{\T}\cdot p(h_T(v)).
$
A pair $(x,y)\in \Delta_k \times \Delta_V$ is said to be a Nash equilibrium if no player can unilaterally improve their profit or cost, that is, 
$$
p(x,y) \ge p(x',y)   \quad \text{for all } x'\in \Delta_k \qquad \text{and} \qquad p(x,y) \le p(x,y')  \quad \text{for all } y'\in \Delta_V.
$$

Von Neumann's minimax  theorem~\cite{v1928theorie}, a basic game-theoretic fact, states that a pair $(x,y)$ is a Nash equilibrium if and only if $x$ defines an optimal solution to the LP,
$$ \text{[P]}  \qquad  \max\left\lbrace  z \,:\,
    \ z\le  \sum_{\T \in \mathcal{T}_k}x_{\T} \cdot p(h_T(v))  \text{ for all }  v \in V \text{, and } x \in \Delta_k\right\rbrace,
$$
and $y$ is an optimal solution to
$$ \text{[D]} \qquad  \min\left \lbrace t \,:\,
 \ t \ge \sum_{v \in V}y_{v} \cdot p(h_T(v)) \text{ for all } T\in \mathcal{T}_k\text{, and } \ y\in \Delta_V\right\rbrace. 
$$
Observe that [P] and [D] are dual LPs, and hence they attain the same optimal value. We refer to this value as $u^*$, the \textit{optimal profit} or the \textit{value of the game}. Moreover, as in an optimal solution for [P] the value of $z$ is simply the minimum value of $\sum_{\T \in \mathcal{T}_k}x_{\T} \cdot p(h_T(v))$ over all $v\in V$, then [P] is a restatement of our original problem in Equation~\eqref{eq:arbitraryValue}. For a given probability vector $x\in \Delta_k$ we say that $\min_{v \in V}\sum_{\T \in \mathcal{T}_k}x_{\T} \cdot p(h_T(v))$ is its objective value. Similarly, for a vector $y\in \Delta_V$ we say that $\max_{T\in \mathcal{T}_k}\sum_{v \in V}y_{v} \cdot p(h_T(v))$ is its objective value.

\section{Search on a line}\label{sec:line}
In this section we focus on the special case where $G=(V,E)$ is a line, that is, $V=\{0,1,\ldots,n-1\}$ and $E=\{\{v,v+1\}:0\le v\le n-2\}$. Also, we restrict ourselves to unit profit functions, where $p(t)=1$ for all $t\le k$ and $p(t)=0$ if $t>k$. To avoid trivial border cases, we assume that $k\ge 2$. Moreover, we can assume that $n > 2^k$, as otherwise there exists a single search strategy, namely the usual binary search, that finds the target in every single node, and hence the value of the game is trivially 1. The aim of this section is to compute a highly structured Nash equilibrium.

We start by characterizing the sets $C\subseteq V$ that arises as covered sets of a search strategy $T\in \mathcal{T}_k$. Then, we restrict the set $\mathcal{T}_k$ to a smaller set of strategies, which we call \textit{efficient}. Using our characterization of covered sets together with LP duality, we will able show that there exists an optimal solution for the seeker that only uses such strategies.

\paragraph*{Covered Sets and Efficient Strategies}
For a given search strategy $\T$, recall that $C(T)$ denotes the set of covered nodes. Most of the technical work we do is on showing that an randomized optimal solution only selects search strategies where $C(T)$ corresponds to a set of consecutive nodes modulo~$n$. For $u,v\in \mathbb{N}$, let $[u,v]_r$ denote the set $\{w \mod r: u\le w\le v\}$ if $u\le v$, and $[u,v]_r=\emptyset$ otherwise. We say that $[u,v]_r$ is an \textit{interval modulo $r$}, or simply an \textit{interval}. To avoid unnecessary notation we omit the subscript $r$ when $r=n$, i.e., $[u,v]=[u,v]_n$. Moreover, for an integer $\ell \le r$, we denote by $[v\oplus \ell]_r = [v,v+\ell-1]_r$, that is, the interval that starts at $v$ of length $\ell$.

It will be particularly useful to understand the sets that appear as a cover set of a search strategy $\T\in \mathcal{T}_k$.
Let $c:=2^k-2$. Observe that a search strategy that aims to cover a single interval in $[1,n-2]$, is able to cover an interval of length $c$: just apply a binary search restricted to a given interval. On the other hand, if the strategy covers a single interval that contains either $0$ or $n-1$ (or both), it will be able to cover $c+1$ nodes. Notice also that there cannot be a strategy that covers exactly $[1,c]$, as this immediately implies that $0$ is covered, similarly with $[n-c-1,n-2]$.  Strategies that cover more than one interval cover a smaller number of nodes. For example, if it covers two disjoint intervals (that do not contain $0$ or $n-1$), the number of covered nodes will be $c-1$. The next proposition formalizes this intuition by giving a characterization of maximal cover sets $C(T)$ (that is, the ones such that there is no $T'\in\mathcal{T}_k$ with $C(T)\subset C(T')$) for search strategies $T\in \mathcal{T}_k$.

\begin{proposition}\label{prop:non-dominated} Denote $c=2^k-2$. Let $C= [u_1 \oplus \ell_1] \cup \ldots \cup [u_s \oplus \ell_s] \subseteq V$ be a set
where $0\le u_1< \ldots < u_s\le n-1$ (and $s$ is minimal).
The set $C$ is a maximal covered set $C(T)$ for some $T \in \BSTs_k$ if and only if
\begin{enumerate}[i)]
    \item $u_1\neq 1$ and $u_s+\ell_s-1\neq n-2$,
    \item $(u_{t+1} - (u_t +\ell_t)) \geq 2,$ for all $t\in \{1,\ldots,s-1\}$ and $u_1- (u_s + \ell_s - n)\ge 2,$ and
    \item $\sum_{t=1}^s \ell_t = \begin{cases}c+1 - s \text{ if } \{0,n-1\}\cap C =\emptyset,\\ c+2 - s \text { otherwise.}\end{cases}$
\end{enumerate}
\end{proposition}

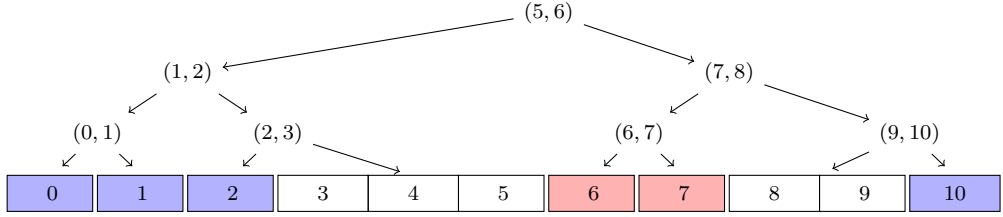
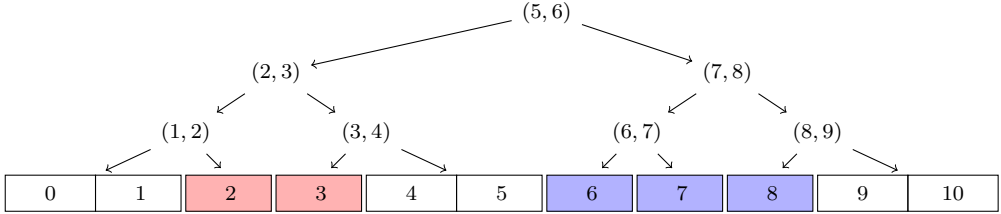
\begin{figure}[htb]
    \centering
    \begin{subfigure}[t]{\textwidth}
        \centering
        \begin{tikzpicture}[yscale=0.8,xscale=0.9,font=\scriptsize]  

\tikzset{
  q/.style={very thick, black},
  nq/.style={dashed, black},
  g/.style={fill=gray!50, draw=none},
  interval1/.style={fill=red!30},
  interval2/.style={fill=blue!30},
  w/.style={},
  failed/.style={->}
}

\begin{scope}[]
    \foreach \i / \s / \n / \sep in {0/nq/interval2/0.07, 1/q/interval2/0.07, 2/q/interval2/0.07, 3/q/w/0, 4/nq/w/0, 5/nq/w/0.07, 6/q/interval1/0.07, 7/q/interval1/0.07, 8/q/w/0, 9/nq/w/0.07, 10/q/interval2/0}
    {   
        \draw[\n] (\i*4/3, -0.3) rectangle (\i*4/3 + 4/3 - \sep, 0.3); 
        \node () at (\i*4/3 + 2/3, 0) {\i}; 
        \node[circle] (\i) at (\i*4/3 + 2/3, 0.29) {}; 
        \node[circle] (-\i) at (\i*4/3, 0.31) {}; 
    };
\end{scope}

\begin{scope}[every node/.style={}]
    \node (Q1) at (6*4/3, 3) {$(5, 6)$};
    \node (Q2) at (2*4/3, 2) {$(1, 2)$};
    \node (Q3) at (8*4/3, 2) {$(7, 8)$};
    \node (Q4) at (4/3, 1) {$(0, 1)$};
    \node (Q5) at (4, 1) {$(2, 3)$};
    \node (Q6) at (7*4/3, 1) {$(6, 7)$};
    \node (Q7) at (10*4/3, 1) {$(9, 10)$};
\end{scope}    

\begin{scope}[every edge/.style={->, draw}, every node/.style={fill=white, circle, scale=0.8}]
    \path (Q1) edge (Q2);
    \path (Q1) edge (Q3);
    \path (Q2) edge (Q4);
    \path (Q2) edge (Q5);
    \path (Q3) edge (Q6);
    \path (Q3) edge (Q7);
    \path (Q4) edge (0);
    \path (Q4) edge (1);
    \path (Q5) edge (2);
    \path[failed] (Q5) edge (4);
    \path (Q6) edge (6);
    \path (Q6) edge (7);
    \path[failed] (Q7) edge (-9);
    \path (Q7) edge (10);
\end{scope}

\end{tikzpicture} 
        \caption{Maximal covered set with two intervals, one of which intersects with vertex 0 and vertex 10. Interval $[6 \oplus 2]$ is colored red and interval $[10\oplus 4]$ is colored blue. The total number of covered vertices is 6.}
        \label{fig:border}
    \end{subfigure}

    \vspace{0.5cm}
    
    \begin{subfigure}[t]{\textwidth}
        \centering
        \begin{tikzpicture}[yscale=0.8,xscale=0.9,font=\scriptsize]  

\tikzset{
  q/.style={very thick, black},
  nq/.style={dashed, black},
  g/.style={fill=gray!50, draw=none},
  interval1/.style={fill=red!30},
  interval2/.style={fill=blue!30},
  w/.style={},
  failed/.style={->}
}

\begin{scope}[]
    \foreach \i / \s / \n / \sep in {0/nq/w/0, 1/nq/w/0.07, 2/q/interval1/0.07, 3/q/interval1/0.07, 4/q/w/0, 5/nq/w/0.07, 6/q/interval2/0.07, 7/q/interval2/0.07, 8/q/interval2/0.07, 9/q/w/0, 10/nq/w/0}
    {   
        \draw[\n] (\i*4/3, -0.3) rectangle (\i*4/3 + 4/3 - \sep, 0.3);
        \node () at (\i*4/3 + 2/3, 0) {\i}; 
        \node (\i) at (\i*4/3 + 2/3, 0.25) {};
        \node (-\i) at (\i*4/3, 0.30) {}; 
    };
\end{scope}

\begin{scope}[every node/.style={}]
    \node (Q1) at (6*4/3, 3) {$(5, 6)$};
    \node (Q2) at (4, 2) {$(2, 3)$};
    \node (Q3) at (8*4/3, 2) {$(7, 8)$};
    \node (Q4) at (2*4/3, 1) {$(1, 2)$};
    \node (Q5) at (4*4/3, 1) {$(3, 4)$};
    \node (Q6) at (7*4/3, 1) {$(6, 7)$};
    \node (Q7) at (9*4/3, 1) {$(8, 9)$};
\end{scope}    

\begin{scope}[every edge/.style={->, draw}, every node/.style={fill=white, circle, scale=0.9}]
    \path (Q1) edge (Q2);
    \path (Q1) edge (Q3);
    \path (Q2) edge (Q4);
    \path (Q2) edge (Q5);
    \path (Q3) edge (Q6);
    \path (Q3) edge (Q7);
    \draw[failed] (Q4) -> (-1);
    \path (Q4) edge (2);
    \path (Q5) edge (3);
    \draw[failed] (Q5) -> (-5);
    \path (Q6) edge (6);
    \path (Q6) edge (7);
    \path (Q7) edge (8);
    \draw[failed] (Q7) -> (-10);
\end{scope}

\end{tikzpicture} 
        \caption{Maximal covered set with 2 intervals. Interval $[2 \oplus 2]$ is colored red and interval $[6\oplus 3]$ is colored blue. Even though the number of intervals is also 2, the number of covered vertices is 5 because neither vertex 0 nor vertex 10 is covered.}
        \label{fig:noborder}
    \end{subfigure}
    \caption{An example of two maximal covered sets and their respective search trees on a line of length $n=11$ and a budget of $k=3$ queries. Covered vertices are colored, and dashed arrows point to leaves of the search tree with two or more vertices. Figure \ref{fig:border} shows the first case of condition (iii) of Proposition \ref{prop:non-dominated} and Figure \ref{fig:noborder} shows the second case.}
    \label{fig:proposition1}
\end{figure}

The first two conditions in the proposition encompass the fact that the vertices not covered by a search strategy always have at least one neighbor that is also not covered. 
The third condition says that there is a trade-off between the number of intervals needed in the description of $C$ and the number of vertices in $C$. Finally, it also says that if $C$ contains $0$ or $n-1$ (or both), we gain one extra covered node. See Figure \ref{fig:proposition1} for an illustration.

Of particular interest are search strategies with a single interval, i.e., with $s=1$, as they maximize the number of covered elements. We call such strategies \emph{efficient}. The following corollary is a direct consequence of Proposition~\ref{prop:non-dominated}.

\begin{corollary}\label{cor:efficient}
For every $v \in V\setminus\{1\}$ there exists a search strategy $\T_v$ that covers,
$$
    C(\T_v)= \begin{cases}
        [v\oplus (c+1)] & \text{if } \{0,n-1\}\cap  [v\oplus (c+1)]\neq \emptyset,\\
        [v\oplus c] & \text{otherwise}.
    \end{cases}  
$$
\end{corollary}

Observe that there is no efficient strategy (of length $c$) that covers an interval starting at~$1$, as we could then obtain a strategy of length $c+1$ that also covers $0$. For any $v \in V\setminus\{1\}$, we denote by $\T_v$ the search strategy that has a covering set as in the corollary.

\paragraph*{The Seeker's Strategy.}

We will start by constructing a strategy $x=(x_T)_{T\in \mathcal{T}_k}\in \Delta_k$ for the seeker given by a greedy algorithm that defines the strategies of $x$. Given $\mathcal{X}$ which denotes the support of $x$, we will define the probability vector $x$ uniformly over $\mathcal{X}$, i.e., $x_T = \frac{1}{|\mathcal{X}|} $ if $ T\in \mathcal{X}$ and 0 otherwise.

Intuitively, our aim is to cover all nodes as evenly as possible with efficient strategies, that is, by the same number of strategies in $\mathcal{X}$. Consider for example the case $n=12$ and $k=3$ (i.e., $c=6$), depicted in Figure~\ref{fig:coprime}. Imagine that we choose $\T_0$ to be in $\mathcal{X}$. This alone yields an objective of 0, and hence a natural choice is to consider $\mathcal{X}=\{\T_0,\T_7\}$. However, this implies that nodes in $\{0,1\}$ are covered twice while the rest is covered only once. This imbalance suggests we should add more strategies, the most natural being $\T_2$, yielding an imbalance on nodes in $\{8,9,10,11\}$. By continuing adding strategies to $\mathcal{X}$ in a greedy manner, we obtain the solution in Figure~\ref{fig:coprime}, yielding an objective value of $\min_v \sum_{T: v\in C(T)} x_T=5/9$. These examples suggests Algorithm~\ref{alg:greedySeeker}, that can be interpreted as a greedy rectangle  packing algorithm, see Figure \ref{fig:not-coprime} and Figure \ref{fig:coprime}.
\begin{algorithm}[htb!]
    \SetAlgoLined
     Set $v=c+1$ and $\mathcal{X} =\{\T_0\}$;\\
     \While{$v\not\in\{0,1\}$}{
        $\mathcal{X} = \mathcal{X}\cup \{\T_v\}$\\
        $v = (v+c) \mod n-1.$\label{line:cUpdate}
    }
    Set $x_T=1/|\mathcal{X}|$ if $T\in \mathcal{X}$ and $x_T=0$ otherwise.
\caption{Greedy algorithm for the seeker's strategy}
\label{alg:greedySeeker}
\end{algorithm}

Observe that in Line~\ref{line:cUpdate} of Algorithm~\ref{alg:greedySeeker} we take $v+c$ modulo $n-1$ instead of modulo $n$, which might be counter-intuitive. The reason for this choice is that if we take away node $n-1$, then the cardinality (length) of each set $C(T_v)\setminus\{n-1\}$ for every $v\not\in \{0,1\}$ is $c$. This avoids the case distinction of Corollary~\ref{cor:efficient}..

Given a set $\mathcal{X}$, the objective value of $x$ is $\min_v \sum_{T: v\in C(T)} x_T$. If we reach the case during the while loop where $v=0$, then the probability of covering vertex $v^*$ is $
|\{\T\in \mathcal{X}: v^*\in C(\T) \}|/|\mathcal{X}|,
$ 
which is the same for every $v^*$ as each element is covered by the same number of search strategies in $\mathcal{X}$. This is the reason why $v=0$ terminates the while loop. If we reach the situation where $v=1$ in the while loop, then we should also stop since $\T_1$ is not well defined. We observe that all nodes, except maybe for node $0$, are covered the same number of times by $\mathcal{X}$, i.e., $|\{\T\in\mathcal{X}: 1\in C(\T) \}|=|\{\T\in \mathcal{X}: v\in C(\T) \}|$ for all $v\in V\setminus\{0\}$. We define this quantity as $h:= |\{\T\in\mathcal{X}: 1\in C(\T) \}|$ and we say that it corresponds to the \textit{height} of~$\mathcal{X}$. Similarly, we denote the cardinality of $\mathcal{X}$ as $w=|\mathcal{X}|$. Observe that with these parameters each node, except maybe for node $0$, is covered by $h$ many sets $C(T_v)$ for $\T_v\in \X$. Hence, the objective value of solution $x$ is $\frac{h}{w}$. The following lemma yields an explicit relationship between $w$ and $h$, and shows how to compute these values with a faster algorithm. In particular, it relates $w$ and $h$ with the $\text{gcd}(c,n-1)$, the greatest common divisor of $c$ and $n-1$, and Bezout's identity. Running times in the next lemma are considered in the RAM model.

\begin{lemma}\label{lm:bezout}
Algorithm~\ref{alg:greedySeeker} terminates in finite time. If $x$ denotes the output of Algorithm~\ref{alg:greedySeeker}, then its objective value is $\frac{h}{w}$, where $w$ is the cardinality and $h$ is the height of $\mathcal{X}$. Moreover, $h,w\in\mathbb{N}$ are the numbers that satisfy
\begin{equation}
\label{eq:gcd}
h(n-1)-wc= \begin{cases} 1 & \text{if } \text{gcd}(c,n-1)= 1 \text{ and},\\
0 &\text{otherwise,}
\end{cases}
\end{equation}
where $0<w\le n-2$ is minimal.
In particular, we can compute $w,h$, and the objective value of $x$ in $O(\log n)$ time. Moreover, if $c$ and $n-1$ are not coprime the objective value simplifies to $\frac{h}{w}=\frac{c}{n-1}$ and $h=\frac{c}{\text{gcd}(c,n-1)}$ and $w=\frac{n-1}{\text{gcd}(c,n-1)}$.
\end{lemma}
\begin{proof}
We have already argued about the objective value of $x$. 

We now give a proof for Equation~\eqref{eq:gcd} and that the algorithm terminates. For a given iteration, consider the state of the algorithm right after running line 4, in particular we have a set $\X$ together with an updated value of $v$. Let $w$ be  the cardinality of $\X$ and $h$ its height, defined as $h=\min_{v\in V} |\{T\in X: v\in C(T)\}|$. 

We use a volume argument. Define the volume of $\X$ as $\text{vol}(\X)= \sum_{\T \in \X} |C(\T)\setminus\{n-1\}|$. Observe that the length of all intervals $C(\T)\setminus \{n-1\}$ for $\T\in \X\setminus \{T_0\}$ is $c$. Thus,
$
\text{vol}(\X) = 1 + wc. 
$
On the other hand, $\text{vol}(\X) = (n-1)h + v$. Notice that $v$ can also be seen as $\text{vol}(\X) \mod n-1$, and that the algorithm terminates the first time that $v\in\{0,1\}$. Hence, if $\X$ is the final set, then $w$ is the smallest number such that 
$$ h(n-1)- wc =1 - v,$$
for some $v\in \{0,1\}$. 

Recall that basic properties of modular arithmetic and Bezout's identity~\cite{bezout1779theorie}, implies that all numbers of the form $s(n-1) + tc$ for $s,t\in \mathbb{Z}$ are multiples of $\text{gcd}(c,n-1)$. Moreover, there exists a pair $(s,t)\in \mathbb{Z}^2$ that satisfies the equation $s(n-1) + tc= \text{gcd}(c,n-1)$, and any pair $(s',t')\in \mathbb{Z}^2$ that satisfies this equation is of the form $(s',t')=(s-r\frac{c}{\text{gcd}(c,n-1)},t+r\frac{n-1}{\text{gcd}(c,n-1)})$ for some $r\in \mathbb{Z}$. Therefore, the pair $(s,t)$ that satisfies this equation with largest $t<0$ value satisfies that $|t| \le \frac{n-1}{\text{gcd}(c,n-1)}$, where equality holds if and only if $c$ divides $n-1.$ 

Assume that $\text{gcd}(c,n-1)=1$. With the previous discussion, the smallest $w$ such that $h(n-1)- wc =1$, satisfies that $w< n-1$. This in particular implies that the algorithm terminates after at most $n-1$ iterations. Moreover, the pair $(h',w')$ with the smallest $w'> 0$ such that $h'(n-1)- w'c =0$ implies that $w'=n-1$ (and $h'=c$), as otherwise $c$ and $n-1$ would share a common divisor. Hence, the smallest value of $w$ such that $h(n-1)-wc\in \{0,1\}$ is attained when $h(n-1)-wc=1$, and thus Equation~\eqref{eq:gcd} holds. 

If $\text{gcd}(c,n-1)>1$, then there are no values of $h,w\in \mathbb{Z}$ such that $h(n-1)-wc=1$, and hence $w$ is the smallest number such that $h(n-1)-wc=0$ by construction. This implies Equation~\eqref{eq:gcd} and that the algorithm terminates after $w=\frac{n-1}{\text{gcd}(c,n-1)}$ many iterations.

To obtain a fast algorithm to compute $h$ and $w$, use the Extended Euclidean Algorithm to compute values of $s,t\in \mathbb{N}$ such that $s(n-1)-tc = \text{gcd}(c,n-1)$ and $t>0$ is minimal. This  takes $O(\log n)$ time (in the RAM model)~\cite[Chapter 31]{cormen_introduction_2009}. If $\text{gcd}(c,n-1)=1$ then we are done as $h=s$ and $w=t$. Otherwise, $h=\frac{c}{\text{gcd}(c,n-1)}$ and $w=\frac{n-1}{\text{gcd}(c,n-1)}$.
\end{proof}

With this lemma, we observe that we can perform a search following $x$ in logarithmic time, even without the need of running Algorithm~\ref{alg:greedySeeker}: simply compute $h$ and $w$ with the Extended Euclidean Algorithm and sample the $t$-th element in $\mathcal{X}$ uniformly.

\begin{corollary}\label{cor:efficientAlg}
Let $x$ be the output of Algorithm~\ref{alg:greedySeeker}. We can sample a tree $\T$ with probability $x_{\T}$ in time $O(\log n)$ and using $O(\log n)$ random bits, without the need to compute $x$. After choosing $\T$, each query of the tree can be determined in constant time. 
\end{corollary}

The proof of this result is provided in the Appendix. In order to show that the constructed solution $x$ is indeed optimal, we construct a dual $y$ with the same objective value. Hence, by weak duality this implies that the pair $(x,y)$ is a Nash equilibrium. As suggested by the previous lemma, we should distinguish whether $n-1$ and $c$ are coprime or not. 

Although our construction of the dual solution does not explicitly rely on complementary slackness, the intuition for our approach came from thinking about certain constraints that this concept imposed on the dual variables. In particular, observe that the solution of the seeker induces ``segments'' of vertices that are covered by the same subset of search strategies in $\mathcal{X}$. See, for example, vertices 3 and 4 in Figure \ref{fig:not-coprime} or vertices 5 and 6 in Figure \ref{fig:coprime}. There are $w$ of these segments (ignoring vertex 0 in the not coprime case), and complementary slackness implies that each of them should have a mass of $1/w$ in an optimal solution for the hider (albeit it does not imply that the mass should be uniformly distributed within each segment). The solution we construct shares this structure: the line is split into $w$ segments, and each of them gets a probability of $1/w$ in total. How these segments are chosen and how the mass is distributed within them differs significantly depending on whether $n-1$ and $c$ are coprime or not. However, in both cases the solution has the property that any pair of disjoint intervals covers at most the same probability mass as an interval obtained by gluing them together according to Proposition \ref{prop:non-dominated}. This then implies that the best-response strategy of the seeker is achieved with an efficient strategy, for which the covered mass will be precisely $h/w$.

We start with the not coprime case, which turns out to admit a simpler proof since the segments can be chosen in a more regular manner. On the other hand, the coprime case is significantly more challenging from a technical point of view as $y$ cannot be constructed so regularly.

\paragraph*{Optimality if $n-1$ and $c$ are not coprime.}

If $n-1$ and $c$ are not coprime, Lemma~\ref{lm:bezout} implies that the objective value of the solution $x$ computed by Algorithm~\ref{alg:greedySeeker} equals $\frac{c}{n-1}$. In what follows we will explicitly define a solution $y$ for the dual [P] with the same objective function, thus implying optimality of $x$ and $y$. The dual solution that we will analyze is defined as
\begin{equation}\label{eq:yCoprime}
    y_v = \begin{cases}
        0 & \text{if $v=k\cdot d$ for some $k\in \mathbb{N}_0$,}\\
        \frac{1}{w(d-1)} & \text{otherwise,} 
    \end{cases}
\end{equation}
for all $v\in V$, where $w = (n-1) / d$ and $d=\text{gcd}(c,n-1)$. We can interpret this solution by thinking of the vertices $\{1, \ldots, n-1\}$ as divided into $w$ segments, each containing vertices $kd+1, \ldots, kd+d$ for $k=0 \ldots, w-1$. Within each segment, its probability of $1/w$ is distributed uniformly among its first $d-1$ vertices (see Figure \ref{fig:not-coprime}). Although this interpretation may seem somewhat contrived, it serves to highlight the connection to the construction of the coprime case, which relies heavily on the use of segments.

In what follows, we will show that $y$ is indeed a probability distribution and achieves an objective value of $\frac{c}{n-1}$.

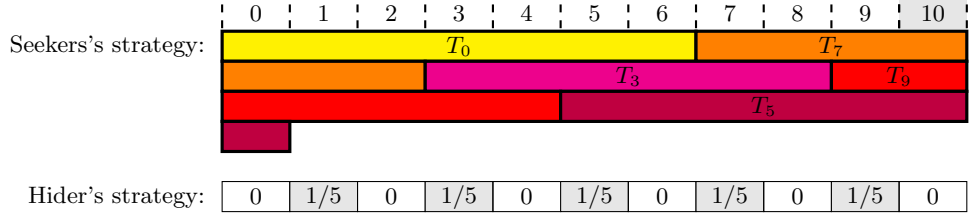
\begin{figure}

\tikzset{
  q/.style={ultra thick, black},
  nq/.style={dashed, thick, black},
  g/.style={fill=gray!50, draw=black, very thick},
  final/.style={fill=gray!20, draw=black, very thick},
  w/.style={fill=none, draw=none, minimum width=1cm},
  r/.style={fill=gray!20, draw=none, minimum width=1cm, minimum height=1cm}
}

\begin{tikzpicture}[yscale=0.4,xscale=0.9,font=\smaller]     
    \draw[r] (9, 0.1) rectangle (10,1.1);
    \foreach \i / \s / \n  in {0/nq/w, 1/nq/w, 2/nq/w, 3/nq/w, 4/nq/w, 5/nq/w, 6/nq/w, 7/nq/w, 8/nq/w, 9/nq/w, 10/nq/w} 
    {   
        \node[\n] at (\i - 0.5, 0.6) {\i};
	    \draw[\s] (\i, 0.1) -- (\i, 1.1); 
    };
    \draw[nq] (-1,0.1) -- (-1, 1.1);
    \draw[nq] (9,0.1) -- (9, 1.1);
    \draw[g,fill=yellow] (-1,-1) rectangle node {$\T_0$} (6, 0);
    \draw[g,fill=orange] (6,-1) rectangle node {$\T_7$} (10, 0); \draw[g,fill=orange] (-1,-2) rectangle  (2, -1);
    \draw[g,fill=magenta] (2,-2) rectangle node {$\T_3$}(8, -1);
    \draw[g,fill=red] (8,-2) rectangle node {$\T_9$} (10, -1); \draw[g,fill=red] (-1,-3) rectangle (4, -2);
    \draw[g,fill=purple] (4,-3) rectangle node {$\T_5$} (10, -2); \draw[g,fill=purple] (-1,-4) rectangle (0, -3);

    \foreach \i in {0, ..., 10} {
        \ifthenelse{\isodd{\i}}{
            \draw[fill=gray!20] (\i -1, -6) rectangle node{$1/5$} (\i, -5);
        }{
            \draw (\i -1, -6) rectangle  node{$0$}  (\i, -5);
        };  
    };

    \draw (-1.1, -5.5) node[left] {Hider's strategy:};
    \draw (-1.1, -0.5) node[left] {Seekers's strategy:};

\end{tikzpicture} 
\caption{Example of the game on a line of length $n=11$ with $k=3$ queries, where $\text{gcd}(c,n-1)=2$. The searcher choose uniformly one out of 5 strategies. Each strategy is depicted with a distinct color marking the covered cells. Here all cells are covered with probability $3/5$, except the first cell which gets more coverage. By appropriately shifting the strategies the seeker can choose to overcover exactly one of the nodes 0,2,4,6,8 or 10. These are exactly the nodes which the hider avoids, choosing uniformly among all other nodes. The value of the game is $3/5$.}
\label{fig:not-coprime}
\end{figure}

\begin{theorem}
    Assume that  gcd$(c,n-1) = d > 1$. Then, $y := (y_v)_{v\in V}$, as defined in Equation~\eqref{eq:yCoprime}, is a feasible solution for [D] and has objective value $\frac{c}{n-1}$. Hence, the pair $(x,y)$ is a Nash equilibrium where $x$ is the solution output by Algorithm~\ref{alg:greedySeeker}.
    \label{thm:not-cp}
\end{theorem}

For a set $S\subseteq V$, we denote $y(S)=\sum_{v\in S}y_v$. Observe that by weak duality the objective value of $y$ cannot be smaller than $\frac{c}{n-1}$. Hence it suffices to show that $y(C(T))\le \frac{c}{n-1}$ for all $T\in \mathcal{T}_k$. Also, observe that as $d$ is a divisor of $n-1$, it holds that $y_{n-1}=0$. Hence, it suffices to analyze the set $C(\T)\setminus\{n-1\}$. Recall that $[u,v]_r$ denotes the interval $\{u,\ldots,v\}$ modulo $r$. Also, notice that $C(\T_v)\setminus \{n-1\} = [v\oplus c]_{n-1}$, for all $V\setminus\{0,1\}$ and $C(\T_0)\setminus \{n-1\}=[0\oplus (c+1)]_{n-1}$. 
The next lemma will be useful to compute $y(C(\T))$ if $\T$ is an efficient strategy. 

\begin{lemma}\label{lm:yshift}
    For any $v\in V$ and integer $\ell\le n-1$ it holds that 
    $$  y([v\oplus\ell]_{n-1}) = \frac{\ell-\lfloor \frac{\ell}{d}\rfloor}{w(d-1)}\quad\text{or}\quad y([v\oplus \ell]_{n-1}) = \frac{\ell-\lceil \frac{\ell}{d}\rceil}{w(d-1)}. $$
\end{lemma}
\begin{proof}
    The definition of $y$ implies that for any $v$ and $\ell$ it holds that $y([v\oplus\ell]_{n-1}) = (\ell-m)/(w(d-1)),$ where $m$ is the amount of numbers in $[v\oplus \ell]_{n-1}$ that are multiples of~$d$. It is easy to see that $m= \lfloor \ell/d \rfloor$ or $m= \lceil \ell/d \rceil$. The lemma follows.
\end{proof}

The next technical lemma will be useful to show that  $y(C(\T))$ is maximized when $\T$ is an efficient strategy. In essence, it says that by gluing two intervals into one (and increasing by one the length of the interval, as suggested by Proposition~\ref{prop:non-dominated}), we can only increase the amount of captured probability mass.

\begin{lemma}\label{lm:ymerge}
Consider two disjoint intervals $[u\oplus\ell]_{n-1}$ and $[v\oplus s]_{n-1}$. Then 
$$y([u\oplus \ell]_{n-1})+y([v\oplus s]_{n-1})\le y([u\oplus (\ell+s+1)]_{n-1}).$$
\end{lemma}
\begin{proof}
Observe that 
$$y([u\oplus(\ell+s+1)]_{n-1})= y([u\oplus\ell]_{n-1})+y([(u+\ell)\oplus(s+1)]_{n-1}),$$
and hence it suffices to show that $y([(u+\ell)\oplus(s+1)]_{n-1}) \ge y([v\oplus s
]_{n-1})$. Indeed,
$$
  y([v\oplus s]_{n-1}) \le  y([(u+\ell) \oplus s]_{n-1}) + \frac{1}{w(d-1)} \le y([(u+\ell) \oplus (s+1)]_{n-1}),
$$
 where the first inequality follows by Lemma~\ref{lm:yshift}.
\end{proof}

Now, we have enough tools to show Theorem~\ref{thm:not-cp}. 

\begin{proof}[Proof (Theorem~\ref{thm:not-cp})]

    First we check that $y$ is indeed a probability distribution. To do this, note that there are $(n-1)/d = w$ multiples of $d$ in $V\setminus\{n-1\}$. Hence, we have that
$$y(V)= \frac{n-1-w}{w(d-1)}=\frac{d(n-1)-(n-1)}{wd(d-1)} = \frac{(n-1)}{wd}=1,$$
and hence $y$ defines a probability distribution over $V$.

In order to show that $y$ has an objective value of $\frac{c}{n-1}$, by weak duality it suffices to show that $y(C(T))\le \frac{c}{n-1}$ for all $\T\in \mathcal{T}_k.$ First of all, observe that for any $v\in V$ we have that 
\begin{align*}
    y([v\oplus c]_{n-1})\le \frac{c - \lfloor \frac{c}{d}\rfloor}{w(d-1)}
    = \frac{c- \frac{c}{d}}{w(d-1)}
    = \frac{c(d-1)}{(n-1)(d-1)}= \frac{c}{n-1},
\end{align*}
where the second equality follows as $d$ divides $c$ by definition. Hence, if $v\in V\setminus\{0,1\}$ we have that $y(C(\T_v))= y(C(\T_v)\setminus\{n-1\})= y([v\oplus c]_{n-1})\le \frac{c}{n-1}$. Finally, if $v=0$, as $y_0=0$, then $y(T_0)=y([0\oplus (c+1)]_{n-1})= y([1\oplus c]_{n-1}) \le \frac{c}{n-1}$. We conclude that $y(C(T))\le \frac{c}{n-1}$ for any efficient search strategy $\T$. 

Let $\T \in \mathcal{T}_k$ be any tree with $C(\T)$ maximal. We must show that $y(C(T))\le c/(n-1)$. By Proposition~\ref{prop:non-dominated} we know that $$ C(T) = [u_1 \oplus \ell_1] \cup \ldots \cup [u_s \oplus \ell_s],$$
where $0\le u_1< \ldots < u_s\le n-1$ and $u_1\neq 1$. Assume that $s\ge 2$, as otherwise $\T$ is an efficient strategy and we are done. Observe that $C(T)\setminus\{n-1\} = [u_1 \oplus \ell'_1]_{n-1} \cup \ldots \cup [u_s \oplus \ell'_s]_{n-1},$ where $\ell'_t = \ell_t$ for all $t\le s-1$ and $\ell_s'=\ell_s$ if $n-1\not\in [u_s \oplus \ell_s]$ and $\ell'_s=\ell_s-1$ otherwise. Regardless, reinterpreting condition (iv) in Proposition~\ref{prop:non-dominated}, and recalling that intervals are maximal in $C(\T)$, it holds that if $u_1\neq 0$ then $\sum_{t=1}^s \ell'_t=c+1-s$ and $\sum_{t=1}^s \ell'_t=c+2-s$ if $u_1=0$. 

Also by Proposition~\ref{prop:non-dominated}, there exists a search strategy $\T'\in \mathcal{T}$ such that 
\[C(\T')\setminus\{n-1\}=[u_1 \oplus (\ell'_1+\ell'_2-1)]_{n-1} \cup  [u_{3} \oplus \ell'_{3}]_{n-1}\cup \ldots \cup [u_s \oplus \ell'_s]_{n-1}, \]
that is, we merged the first and second interval, gaining one unit in the total length. Observe that regardless of weather $u_1=0$ or not, the new intervals satisfy the conditions of Proposition~\ref{prop:non-dominated}, guaranteeing the existence of $\T'$. Moreover, by Lemma~\ref{lm:ymerge}, we have that 
$$y(C(\T))=y(C(\T)\setminus\{n-1\})\le y(C(\T')\setminus\{n-1\})=y(C(\T')).$$
Iterating this argument we obtain that $y(C(\T))\le y(C(T_{v_1})).$ The theorem follows by weak duality.
\end{proof}

\paragraph*{The coprime case.}
\begin{figure}[htb]

\tikzset{
  q/.style={ultra thick, black},
  nq/.style={dashed, thick, black},
  g/.style={fill=gray!50, draw=black, very thick},
  stop/.style={fill=cyan, draw=black, very thick},
  final/.style={fill=gray!20, draw=black, very thick},
  w/.style={fill=none, draw=none, minimum width=1cm},
  r/.style={fill=gray!20, draw=none, minimum width=1cm, minimum height=1cm}
}

\begin{tikzpicture}[yscale=0.4,xscale=0.9,font=\smaller]      \draw[r] (10, 0.1) rectangle (11,1.1);
    \foreach \i / \s / \n  in {0/nq/w, 1/nq/w, 2/nq/w, 3/nq/w, 4/nq/w, 5/nq/w, 6/nq/w, 7/nq/w, 8/nq/w, 9/nq/w, 10/nq/w, 11/nq/w} 
    {   
        \node[\n] at (\i - 0.5, 0.6) {\i};
	\draw[\s] (\i, 0.1) -- (\i, 1.1); 
    };
    \draw[nq] (-1,0.1) -- (-1, 1.1);
    \draw[nq] (10,0.1) -- (10, 1.1);

    \draw[g,fill=yellow] (-1,-1) rectangle node {$T_0$} (6, 0);
    \draw[g,fill=orange] (6,-1) rectangle  node {$T_7$} (11, 0); \draw[g,fill=orange] (-1,-2) rectangle  (1, -1);
    \draw[g,fill=magenta] (1,-2) rectangle node {$T_2$} (7, -1);
    \draw[g,fill=red] (7,-2) rectangle node {$T_8$} (11, -1); \draw[g,fill=red] (-1,-3) rectangle (2, -2);
    \draw[g,fill=purple] (2,-3) rectangle node {$T_3$} (8, -2); 
    \draw[g,fill=pink] (8,-3) rectangle node {$T_9$} (11, -2); \draw[g,fill=pink] (-1,-4) rectangle (3, -3);
    \draw[g,fill=violet!80] (3,-4) rectangle node {$T_4$} (9, -3); 
    \draw[g,fill=blue!50] (9,-4) rectangle node {$T_{10}$} (11, -3); \draw[g,fill=blue!50] (-1,-5) rectangle (4, -4);
    \draw[stop] (4,-5) rectangle node {$T_{5}$} (11, -4); 

    \draw (-1.1, -6.5) node[left] {\footnotesize Hider's strategy:};
    \draw (-1.1, -0.5) node[left] {\footnotesize Seeker's strategy:};

    \foreach \i / \s /\c in { 1/9/gray!40, 2/9/gray!40, 3/9/gray!40, 4/9/gray!40, 5/18/gray!20, 6/18/gray!20, 7/9/gray!40, 8/9/gray!40, 9/9/gray!40, 10/9/gray!40} {
            \draw[fill=\c] (\i -1, -7) rectangle node{$1/\s$} (\i, -6);
    }
        \draw (-1, -7) rectangle  node{$0$}  (0, -6);
        \draw (10, -7) rectangle  node{$0$}  (11, -6);
    
\end{tikzpicture} 
\caption{Example of the game on a line of length $n=12$ with $k=3$ queries. The searcher chooses one out of 9 strategies. Each strategy is depicted with a distinct color marking the covered cells. The seeker covers uniformly all cells. However the hider chooses a non-uniform distribution where segments have length 1 or 2. The value of the game is $5/9$.}
\label{fig:coprime}
\end{figure}
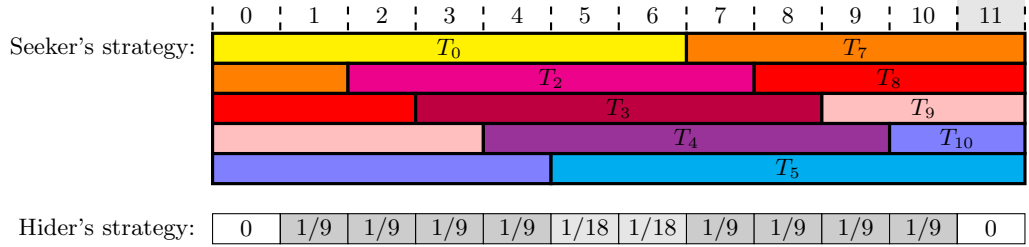

We now focus on instances of the game where gcd$(c,n-1)=1$. By Lemma~\ref{lm:bezout}, we have that the objective value of the solution $x$ constructed by Algorithm~\ref{alg:greedySeeker} is $\frac{h}{w}$, where $h, w\ge 1$ are the natural numbers such that $h (n-1) = w c + 1$, with $w$ minimal. As in the previous case, we define explicitly a solution $y$ and show it is optimal. The structure of the solution is similar to the previous case but more intricate and its analysis more technical.

Ideally, we would like to assign a mass of $\frac{h}{wc}$ to each coordinate $y_v$. In this way, each efficient strategy $\T\in \mathcal{T}_k$ of length $c$ would capture a probability of $\frac{h}{w}$. However, this simple idea must be refined as there are issues with the boundary, where efficient strategies have length $c+1$. To handle this situation we define a function $g: V \rightarrow \mathbb{Q}$, the {\it ideal mass} function, that we will use as a benchmark for the cumulative distribution of the solution for the hider $y$. It is defined as follows:
\begin{equation*}
g(v) = v \cdot \frac{h}{wc}.    
\end{equation*}
The general idea is that, for any vertex $v \in [1,n-2]$, the cumulative mass of the solution $y$ up to $v$ remains bounded by
\begin{equation}
    g(v) \leq y([1,v])< g(v+1).
    \label{eq:target}
\end{equation}
Such a distribution will give us an optimal solution for [D], as we will show later. To satisfy Equation~\eqref{eq:target}, we partition nodes in $[1,n-2]$ into $w$ segments of two possible lengths: $r:=\lfloor\frac{c}{h}\rfloor$ or $r+1$, where length refers to the amount of nodes contained in a segment. Note that $c>h$, so $r$ is at least 1. In each segment, a total probability mass of $\frac{1}{w}$ is distributed uniformly among its nodes. Nodes $0$ and $n-1$ are assigned 0 probability mass. 

The length and order of the segments is chosen in the following way: start at node $v=1$, and choose the largest $r^* \in \{r, r+1\}$ such that 
\begin{equation}
g(v+ r^* -1) \leq y([1, v-1]) + \frac{1}{w},
\label{eq:SegmentRule}
\end{equation}
then, $r^*$ is the length of the segment starting at $v$. Set $y_i = \frac{1}{r^* w}$ for $i \in \{v, \ldots, v+ r^* - 1\}$. Update $v = v + r^*$ and repeat the same selection rule until $v > n-2$. We call Equation (\ref{eq:SegmentRule}) the \emph{segment rule}. Observe that $\frac{1}{r+1} < \frac{h}{c} < \frac{1}{r} $, so the cumulative mass grows faster than $g$ in short segments and slower than $g$ in long segments. Hence, selecting a long segment decreases the difference $y([1, v+r^*]) - g(v+r^*)$ and a short segment increases it. See Figure~\ref{fig:hider_xmpl} for an example. We start by showing that the construction is correct, that is, that for either $r$ or $r+1$ Equation~\ref{eq:SegmentRule} is satisfied.

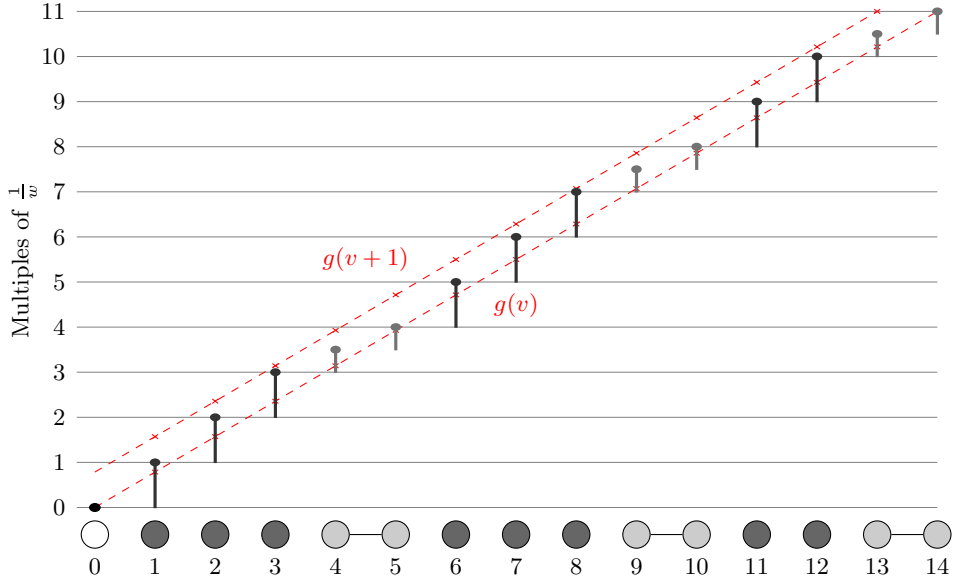
\begin{figure}[htb]
    \centering
    \begin{tikzpicture}[yscale=0.6, xscale=0.8, font=\smaller]
\tikzset{
    short/.style={fill=black!60},
    long/.style={fill=black!20},
    b/.style={black!80},
    g/.style={black!55},
    n/.style={fill=white}
}

%eje y
\foreach \y in {0,1,2,3,4,5,6,7,8,9,10,11}
{   
    \node [label=west:$\y$] at (-0.15, \y) {};
    \draw[gray, very thin] (-0.3, \y) -- (14, \y);
}
%\draw (0, -0.15) -- (0, 11.3);
\node[rotate=90] at (-1.2, 5.5) {Multiples of $\frac{1}{w}$};

%\draw (-0.15, 0) -- (14.3, 0);

%%%%% y([1, v]) %%%%%

% g(v)
\draw[red, dashed] plot[mark=x, mark phase=2] coordinates {(0,0) (1,1*11/14) (2,2*11/14) (3,3*11/14) (4,4*11/14) (5,5*11/14) (6,6*11/14) (7,7*11/14) (8,8*11/14) (9,9*11/14) (10,10*11/14) (11,11*11/14) (12,12*11/14) (13,13*11/14) (14,14*11/14)};
\node[red] at (7, 4.5) {$g(v)$};

% g(v+1)
\draw[red, dashed] plot[mark=x, mark phase=2] coordinates {(0,1*11/14) (1,2*11/14) (2,3*11/14) (3,4*11/14) (4,5*11/14) (5,6*11/14) (6,7*11/14) (7,8*11/14) (8,9*11/14) (9,10*11/14) (10,11*11/14) (11,12*11/14) (12,13*11/14) (13,14*11/14)};
\node[red] at (4.5, 5.5) {$g(v+1)$};

\draw[fill=black] (0, 0) circle (2.5pt);
\foreach \x / \y / \l / \c in {1/1/1/b, 2/2/1/b, 3/3/1/b, 4/3.5/0.5/g, 5/4/0.5/g, 6/5/1/b, 7/6/1/b, 8/7/1/b, 9/7.5/0.5/g, 10/8/0.5/g, 11/9/1/b, 12/10/1/b, 13/10.5/0.5/g, 14/11/0.5/g}
{
    \node[minimum height=1cm] (\x) at (\x, \y+0.8) {};
    \node[minimum height=1cm] (-\x) at (\x, \y-\l-0.86) {};
    \path[draw, very thick, \c] (\x) edge (-\x);
    \fill[\c] (\x, \y) circle (2.5pt); 
}

%\draw plot coordinates {(0,0) (1,1) (2,2) (3,3) (5,4) (6,5) (7,6) (8,7) (10,8) (11,9) (12,10) (14,11)};
%\draw[thick] plot[const plot mark mid] coordinates {(0,0) (1,1) (2,2) (3,3) (4, 3.5) (5,4) (6,5) (7,6) (8,7) (9,7.5) (10,8) (11,9) (12,10) (13, 10.5) (14,11)};

%eje x
%\foreach \x in {1,2,3,4,5,6,7,8,9,10,11,12,13,14}
%   \draw (\x cm,1pt) -- (\x cm,-1pt) node[anchor=north] {$\x$};
%\draw (-0.15, 0) -- (14.3, 0);
\foreach \x / \len in {0/n, 1/short, 2/short, 3/short, 4/long, 5/long, 6/short, 7/short, 8/short, 9/long, 10/long, 11/short, 12/short, 13/long, 14/long}
   \node[draw, thin, circle, \len] (\x) [label=below:$\x$] at (\x, -0.6) {};

\path[draw] (4) edge (5);
\path[draw] (9) edge (10);
\path[draw] (13) edge (14);

\end{tikzpicture}
    \caption{Cumulative distribution of the hider's solution for $n=38$ and $k=4$ ($c=14$) in the first 14 vertices. In this instance, $h=11$ and $w=29$, which means short segments (resp. long) have length 1 (resp. 2). In this figure, connected vertices are in the same segment. The probability assigned to each vertex is illustrated by its shade. Notice that $y([1,v])$ always remains within $g(v)$ and $g(v+1)$.}
    \label{fig:hider_xmpl}
\end{figure}

 \begin{lemma} Let $v$ be the first node of a segment. Then $g(v-1)\le y([1,v-1])$ and Equation~\eqref{eq:SegmentRule} holds for $r^*=r.$
 \end{lemma}
 \begin{proof}
  We show the lemma inductively. For the base case, clearly $g(0)=0\le y([1,0])=0$. Let $v\ge 1$ be the first node of a segment and assume that $g(v-1)\le y([1,v-1])$. Therefore, $g(v+r-1)=g(v-1)+r\frac{h}{wc}\le y([1,v-1]) + r\frac{h}{wc} \le y([1,v-1]) + 1/w$. By construction, regardless of the value of $r^*\in \{r,r+1\}$, Equation~\eqref{eq:SegmentRule} holds. Therefore, as $v+r^*-1$ is the last node of the segment of $v$, it holds that $g(v+r^*-1)\le y([1,v-1])+1/w = y([1,v+r^*-1])$. The induction follows.
 \end{proof}

 The next proposition gives several useful properties of our construction of $y$. 

 \begin{lemma}    
    The constructed solution $y$ satisfies the following properties.
    \begin{enumerate}[a)]
        \item For all $v \leq  n-2$, it holds that $g(v) \leq y([1,v]) < g(v+1).$ \label{it:a}
         \item Node $v\in [1,n-2]$ belongs to the $s$-th segment where 
          $s=\ceil{\frac{vh}{c}}.$ \label{it:b}
        \item $y([1, v]) = g(v)$ if and only if $v$ is a multiple of $c$. If $v$ is a multiple of $c$ then it is the last node of a segment.\label{it:c}
        \item $y_v=y_{v+c}$ for all $v\in [1,n-2-c]$. \label{it:d}
        \item The construction of $y$ finishes with $w$ segments. The last node of the last segment is $n-2$. \label{it:e}
    \end{enumerate}
    \label{lm:yProperties}
\end{lemma}

\begin{proof}
     To show a), observe that the inequality $g(v-1) \leq y([0, v-1])$ holds trivially for $v=1$. We proceed by induction. Let $v \geq 1$ be the first node of a segment of length $r^*$. By the previous lemma we know that $g(v-1)\le y([1,v-1])$ and that $g(v+r^*-1)\le y([1,v+r^*-1])$. Hence,  $g(v+r') \leq y([1, v+r'])$ also holds for $v+r'$ with $0 \le r' < r^*-1$: it is true for $r'=-1$ and for $r'=r^*$ and $g(v+r')$ and $y([1, v+r'])$ both grow linearly in said interval. 
    
    Now we focus on the right hand side inequality. Again, the case $v=0$ trivially holds. For the inductive step, assume that for some segment that starts at node $v \geq 1$ we have that $y([1, v-1]) < g(v)$. We have two options: if the segment has length $r^*=r+1$, then it is clear that $y([1, v+r']) < g(v + r')$ for all $0\le r' \leq r+1$: Indeed, in said range $g(\cdot)$ grows at a rate of $\frac{h}{wc}$ while $y([1,\cdot])$ grows at a rate of $\frac{1}{(r+1)w} < \frac{h}{wc}$. 
    
    If instead the length of the segment is $r$, we know that Equation \eqref{eq:SegmentRule} is violated by $r^*=r+1$, i.e. 
    $$
        y([1, v+r-1]) = y([1, v-1]) + \frac{1}{w} < g(v+r).
    $$
    With the same argument as before, the inequality must also be satisfied for $v + r'$, with $0 < r' < r$.

    To show b), let $v \in V$ be the final vertex of the $s$-th segment, for $s\ge 1$. Using a), we get that
    $$
    v\cdot \frac{h}{wc} \leq y([1, v]) = \frac{1}{w}\cdot s < (v+1) \cdot \frac{h}{wc},
    $$
    and hence $v\cdot \frac{h}{c} \leq s < (v+1) \cdot \frac{h}{c} < v\cdot \frac{h}{c} +1 $ since $h< c$ as $h(n-1)-wc=1$ and $w\le n-2$. This implies that $s=\ceil{v\cdot \frac{h}{c}}$ and $s+1=\ceil{(v+1)\cdot \frac{h}{c}}$. For the same reasoning, if $v'$ is the last vertex of the $(s+1)$-th segment, then $s+1=\ceil{v'\cdot \frac{h}{c}}$. This implies that for every vertex $u$ in the $(s+1)$-segment, $s+1$ equals $\ceil{u\cdot \frac{h}{c}}$. Hence, b) holds for the $s$-th segments where $s\ge 2$. Finally, if $v'$ is the last node of the first segment, then $1=\ceil{v'\cdot \frac{h}{c}}$. For all $1\le v\le v'$ it also holds that $1\le \ceil{v\cdot \frac{h}{c}}\le \ceil{v'\cdot \frac{h}{c}}=1$, which implies b) for the case $s=1$.

    For c), we use b) to observe that if $v$ is a multiple of $c$ then $\ceil{v\frac{h}{c}}= v\frac{h}{c}$ and thus $v$ must be the last node of the $(v\frac{h}{c})$-th segment, as $\ceil{(v+1)\frac{h}{c}} = v\frac{h}{c} +1$. Thus $y([1,v])= \frac{1}{w}\cdot(v\frac{h}{c})= g(v).$ 
    
    For the other implication, notice that as $g(u)\le y([1,u])$ for all $u$, if $g(v)=y([1,v])$ then $v$ must be the last node of a long segment. Hence, by b),
    $$g(v)= v \frac{h}{cw} = y([1,v])=\frac{1}{w}\cdot \ceil{v\frac{h}{c}},$$
    which implies that $\frac{vh}{c}$ must be integral. As by Bezout's identity  $h$ and $c$ are coprime (since a common divisor of $h$ and $c$ would also be a divisor of 1), this implies that $v$ is a multiple of $c$. 

    To show d), consider first the case that $v \in \{1, \ldots, (n-2)-c\}$. Observe that node $v+c$ is in segment number
    $$
    \ceil{\frac{(v+c)h}{c}} = \ceil{\frac{vh}{c}} + h.
    $$
    As this holds for every $v$, we conclude that $v$ is in a segment of length $r^*$ if and only if $v+c$ is in a segment of length $r^*$. Hence, $y_v=y_{v+c}$.
    
    Consider now e). By b), node $n-2$ is in segment
    $$
    \ceil{\frac{(n-2)h}{c}} = \ceil{\frac{wc + (1-h)}{c}} = \ceil{w-\frac{(h-1)}{c}} = w,
    $$
    where the last equality follows as $(h-1)/c < 1$. Additionally,
    $$
    \ceil{\frac{(n-1)h}{c}} =  \ceil{w+\frac{1}{c}} = w + 1,
    $$
    so $n-2$ is the end of the $w$-th segment.
\end{proof}

In what follows, we introduce two fundamental results that are crucial for our construction. The two of them combined will make it possible to use the gluing argument utilized in the not coprime case.

\begin{lemma}
    The following equalities hold,
    \begin{enumerate}[i)]
        \item $y([v \oplus \ell]) \leq y([1\oplus \ell])$ for all $\ell,v\in [1,c]$, and
        \item $y([v \oplus \ell]) \ge y([t\oplus  \ell])$, for all $\ell,v\in [1,c]$ where $t := (n-1) \mod c$.
    \end{enumerate}
    Moreover, the segment starting at node 1 and the one ending at node $t-1$ have length $r$, and the segment starting at $t$ and the one ending at $c$ have length $r+1$.
    \label{lm:yExtrema}
\end{lemma}

\begin{proof}
        We know that vertices are assigned one of two possible probabilities, depending on the length of the segment in which they are located. In this proof, we call vertices in short segments \emph{heavy} and vertices in long segments \emph{light}. 
    
    Let us start with i). We proceed by induction. For an interval of length $1 \leq \ell \leq r$, the maximum mass attainable is that of $\ell$ heavy vertices. As $g(1+r) > 1/w$, we see that only a short segment can satisfy Equation \eqref{eq:SegmentRule} at $v=1$, so the first segment has length $r$. Then, $y([1\oplus \ell]) \leq y([v\oplus \ell])$ for all $\ell \leq r$ and for all~$v\in [1,c]$.  
    
    Now, assume by contradiction that $\ell > r$ is the smallest length for which there exists a vertex $v \in \{2, \ldots, c-1\}$ such that $y([v\oplus \ell]) > y([1\oplus \ell])$. Without loss of generality, vertex $v$ is the first vertex of a short segment. If this were not the case, simply slide interval $[v \oplus \ell]$ to the closest vertex $v'$ that marks the beginning of a short segment. Then, we have $y[v' \oplus \ell] \geq y[v \oplus \ell]$, since we have at worse traded heavy vertices for other heavy vertices. Continuing, observe that $y([v\oplus \ell]) > y([1\oplus \ell])$ implies that there is at least one more heavy vertex in interval $[v \oplus \ell]$ than in interval $[1 \oplus \ell]$, because both intervals have the same length. By the inductive hypothesis, we also know that
    $$y([1\oplus (\ell-1)]) \ge y([v\oplus (\ell-1)]).$$  
    Hence, by the same token, the number of heavy vertices in $[v \oplus (\ell-1)]$ is at most the same as in $[1 \oplus (\ell - 1)]$. For these two facts to be true at the same time, both intervals must contain the same number of light and heavy vertices, that is,
    \begin{equation}
    y[1 \oplus (\ell-1)] = y[v \oplus (\ell-1)]. 
    \label{eq:lemma3.6}
    \end{equation}
    From this statement we can deduce that $\ell$ is a light vertex and $v + \ell - 1$ is a heavy vertex. Next, we will show that both vertices $\ell$ and $v + \ell - 1$ are, in fact, the first of their corresponding segment. From what we have established so far, we know that vertex $\ell$ is in a long segment and vertex $v + \ell - 1$ is in a short segment. From the inductive hypothesis, we also know that the segment of vertex $\ell-1$ is not shorter than the segment of vertex $v + \ell - 1$. If it were shorter, then the proposition would not hold for length $\ell-2$ either, which contradicts our initial assumption. This leaves three possibilities, and we will check all of them. 
    
    The first option is that vertices $\ell-1$ and $v + \ell - 2$ are both in a short segment. Since vertex $\ell$ is in a long segment, $\ell - 1$ has to be the last vertex of its own segment. Moreover, vertex 1 is the first vertex of a short segment, so interval $[1 \oplus (\ell - 1)]$ contains only complete short segments. Since the length of short segments is $r$, the number of heavy vertices in $[1 \oplus (\ell - 1)]$ has to be a multiple of $r$. We already established that there are the same number of heavy vertices in interval $[1 \oplus (\ell-1)]$ as in interval $[v \oplus (\ell-1)]$, and by assumption $v$ is the first vertex of a short segment. Thus, there are also only complete short segments in $[1 \oplus (\ell-1)]$. We conclude that $v + \ell - 1$ must equally be the last vertex of its segment. 
    
    The second option is that vertices $\ell-1$ and $v + \ell - 2$ are both in a long segment. This case has the same analysis as the previous case, but instead the focus is on the number of light vertices and proving that $\ell-1$ is the last vertex of a long segment. 
    
    The last case is that vertex $\ell-1$ is in a long segment and vertex $v + \ell - 2$ is in a short segment. All short segments in interval $[1 \oplus (\ell - 1)]$ are complete, so as we have done before, we obtain that $[v \oplus (\ell-1)]$ contains a multiple of $r$ heavy vertices, so $v + \ell - 1$ has to be the last vertex of a short segment. The same reasoning applied to the number of light vertices in interval $[v \oplus (\ell-1)]$ shows that $\ell-1$ is the last vertex of a long segment. In any case, we have shown that $v + \ell -1$ and $\ell$ are the first vertices of a short and a long segment, respectively. However, we can also deduce that $v+\ell-1$ is the first vertex of a long segment. In fact,
    \begin{align*}\label{eq:ine1toell}
    g(v+\ell -1 + r) &= g(\ell+r) + g(v-1) \le y([1\oplus (\ell-1)]) + \frac{1}{w} + y([1,v-1]) \\ &= y([v\oplus (\ell-1)]) + \frac{1}{w} + y([1,v-1]) = y([1,v+\ell-2]) + \frac{1}{w},
    \end{align*}
    where the first inequality follows as $\ell$ is the first node of a long segment and Lemma~\ref{lm:yProperties} part \ref{it:a}). The second equality follows by equation \eqref{eq:lemma3.6}. Then, the last equality and the segment selection rule (\ref{eq:SegmentRule}) implies that $v+\ell -1$ is the first vertex of a long segment, which produces a contradiction.

    For ii), the logic is equivalent, but the base case is more technical and the inequalities are switched to the opposite direction. First, we show that vertex $t$ is the first vertex of a long segment. Let $t:=(n-1) \mod c$, or equivalently, 
    $$t  = (n-1) - c \floor{\frac{n-1}{c}}.$$
    With this equality, we check that $t$ is the first vertex of a segment using Lemma~\ref{lm:yProperties} item \ref{it:b}),
    \begin{align*}
        \ceil{ \frac{(t-1)h}{c}} &= \ceil{\left((n-2) - c \floor{\frac{n-1}{c}} \right)\frac{h}{c}} \\
        &= \ceil{\frac{h(n-1)}{c} - \frac{h}{c}} - h \floor{\frac{n-1}{c}} \\
        &= \ceil{\frac{wc + 1}{c} - \frac{h}{c}} - h \floor{\frac{n-1}{c}} \tag{Bezout}\\
        &= \ceil{w - \frac{h-1}{c}} - h \floor{\frac{n-1}{c}} \\
        &= w - h  \floor{\frac{n-1}{c}} \tag{$\frac{h-1}{c} < 1$}.
    \end{align*}
    With an analogous computation we get $\ceil{ \frac{th}{c}} = w+1- h  \floor{\frac{n-1}{c}} > \ceil{ \frac{(t-1)h}{c}}$, and thus $t$ is the beginning of a segment. Moreover, 
    \begin{align}
    y([1,t-1]) &=  \frac{1}{w}\left(w - h  \floor{\frac{n-1}{c}}\right) = \frac{1}{w}\left(w - h  \floor{\frac{n-1}{c}}\right) \notag\\
    & = \frac{h}{cw}\left(\frac{cw}{h} - c\floor{\frac{n-1}{c}}\right) = \frac{h}{cw}\left(\frac{h(n-1)-1}{h} - c\floor{\frac{n-1}{c}}\right) \notag\\
    &= \frac{h}{cw}\left(n-1 - \frac{1}{h} - c\floor{\frac{n-1}{c}}\right) = g(t) - \frac{1}{cw} \label{eq: thin window}.
    \end{align}
    Hence, the segment starting at $t$ is long if and only if $g(t) - \frac{1}{cw} + \frac{1}{w} \ge g(t+r)$, which is equivalent to showing that $rh = \floor{\frac{c}{h}} h \le c-1$, which holds as $c$ and $h$ are coprime. 

    As $t$ belongs to a long segment, we conclude that $y([t\oplus \ell]) \le y([v\oplus \ell])$ for all $\ell \leq r+1$ and for all $v\in [1,c]$, which proves the base case. For the general case, we proceed by contradiction and consider a minimal $\ell > r+1$ for which a vertex $v\in[1,c]$ exists such that $y([t\oplus \ell])> y([v\oplus \ell])$. Without loss of generality, we assume that $v$ is the first vertex of a long segment. Using the same logic as in the proof for i), we can deduce $y([t\oplus( \ell-1)]) = y([v\oplus (\ell-1)])$ and that $v+\ell-1$ and $t+\ell-1$ are the first vertices of a long and a short segment, respectively. Thus, 
    \begin{align*}\label{eq:ine1toell}
    g(t + \ell -1 + r) &= g(v +\ell - 1 +r) + g(t-v) \\
    &\le y([1, v + \ell- 2]) + \frac{1}{w} + g(t)-g(v) \\ 
    &= y([1,v-1]) + y([v\oplus (\ell-1)]) + \frac{1}{w} + g(t) - g(v) \\ 
    &\leq y([t \oplus (\ell-1)]) + \frac{1}{w} + g(t) - \frac{1}{wc} \\
    &=  y([t \oplus (\ell-1)]) + \frac{1}{w} + y([1, t-1])  = y([1, t + \ell - 2]) + \frac{1}{w},\\
    \end{align*}
    where the first inequality follows as $v + \ell - 1$ is the first vertex of a long segment. For the second inequality, we use $y([1, v-1]) \leq g(v) - \frac{1}{wc}$, which follows from part \ref{it:a}) of Lemma \ref{lm:yProperties}  since the largest denominator in the strict inequality is $wc$. The third equality follows from Equation \ref{eq: thin window}. The last equality implies that $t + \ell - 1$ is a light vertex, which is a contradiction. Point ii) follows.

    It only remains to prove that $c$ is part of a long segment and $t-1$ is part of a short segment. Using the left-hand side of part \ref{it:a} of Lemma \ref{lm:yProperties}, and the fact that $c$ is the last vertex of its segment, we have $g(c-r^*) \leq y([1, c]) - \frac{1}{w}$, where $r^*$ is the length of the segment. By replacing part \ref{it:c} of Lemma \ref{lm:yProperties} in the inequality, we get $g(r^*) \geq \frac{1}{w}$. Consequently, we get $r^*=r+1$, that is, vertex $c$ is in a long segment. Similarly for vertex $t-1$, the right-hand side of part \ref{it:a} of Lemma \ref{lm:yProperties}, and the fact that $t-1$ is the last vertex of its segment yield $g(t - r^*) + \frac{1}{wc} \geq y([1, t-1]) - \frac{1}{w}$. By replacing Equation (\ref{eq: thin window}) in the inequality, we get $g(r^*) \leq \frac{1}{w}$, so vertex $t-1$ lies in a short segment and the lemma follows.
\end{proof}

\begin{lemma}\label{lm:ySpread}
    Let $t=n-1 \mod c$. Then for all $\ell \in [1, c]$, we have that
    $$
        y([1\oplus \ell])\leq y([t\oplus (\ell+1)]).
    $$
\end{lemma}
\begin{proof}
    First observe that by periodicity of $y$ the case $\ell=c$ is trivial as $y([1\oplus c])=([t\oplus c])$.
    
    Consider now the case $\ell \leq r$. By Lemma \ref{lm:yExtrema}, node $1$ is the start of a short segment and $t$ is the start of a long segment. This implies that $y([1\oplus \ell]) = \frac{\ell}{w r}$ and that $y([t\oplus (\ell+1)])= \frac{\ell+1}{w (r+1)}$. Hence, for this case it suffices to show that  $\frac{\ell}{r} \le \frac{\ell+1}{r+1}$, which holds as $\ell \le r$.
    
    Similarly, assume that $c-r-1\leq \ell \leq c-1$. In this case $y([1\oplus \ell]) = y([1,c]) - y([\ell+1,c]) = y([1,c]) - \frac{c-\ell}{w(r+1)}$, where the last equality holds as $c$ is the last node of a long segment.

    Observe that that $t-1 \ge r$ as $t-1$ is the last node of a segment by Lemma~\ref{lm:yExtrema}. Hence, by the periodicity of $y$, it holds that $$y([t\oplus (\ell+1)]) = y([1,t+\ell -c]) + y([t,c]) = y([1,c]) - y([t+\ell-c+1,t-1]).$$ As $t-1$ is the last node of a short segment it holds that $y([t\oplus (\ell+1)]) = y([1,c]) - \frac{c-\ell-1}{wr}$. Therefore, for this case it suffices to show that  $\frac{c-\ell-1}{wr}\le \frac{c-\ell}{w(r+1)},$ which holds as $\ell \ge c-r-1$.

    To finish, let $\ell \in [r+1, c-r-2]$. For this case, we will actually prove that $$y([1\oplus \ell])= y([t\oplus (\ell+1)])= y([1,t+\ell])-y([1,t-1]).$$
    
    Consider an integer $1\le \ell'\le c-1$. Let us compute in which segment node $t+\ell'$ belongs to by using Lemma~\ref{lm:yProperties} part \ref{it:b}),
    \begin{align*}
        \ceil{\frac{(t+\ell')h}{c}} &= \ceil{\frac{((n-1) - c \floor{\frac{n-1}{c}}+\ell')h}{c}} = \ceil{\frac{(n-1)h}{c}  +\frac{\ell' h}{c}} -h\floor{\frac{n-1}{c}}\\
        & = \ceil{w +  \frac{ \ell' h+1}{c}} -h\floor{\frac{n-1}{c}} = w+\ceil{\frac{ \ell' h+1}{c}} -h\floor{\frac{n-1}{c}}, 
    \end{align*}
    where the third equality uses that $(n-1)h-cw=1$. Moreover, we observe that as $1\le \ell' \le c-1$ then $\ceil{\frac{\ell'h}{c}} = \ceil{\frac{\ell'h+1}{c}}$. Indeed, if this is not the case, then for $s=\ceil{\frac{\ell'h}{c}}$ it must hold that
    $$\frac{\ell'h}{c}\le s < \frac{\ell'h+1}{c},$$
    and thus $\ell'h \le c\cdot s < \ell'h+1$. As all numbers are integers we get that $\ell'h = c\cdot s$, and thus $s=\frac{\ell'h}{c}$. Since by Bezout's identity $h$ and $c$ are coprime, it must hold that $\ell'$ is a multiple of~$c$, which contradicts the fact that $\ell'<c$. We have shown that for all $\ell'\le c-1$, 
    $$\ceil{\frac{(t+\ell')h}{c}} = \ceil{\frac{ \ell' h}{c}}+  w-h\floor{\frac{n-1}{c}}.$$
    In other words, node $\ell'$ is in the $s$-th segment if and only if $t+\ell'$ is in the $(s+w-h\floor{\frac{n-1}{c}})$-th segment. This implies that $\ell'$ is the last node of a segment if and only if $t+\ell'$ is the last node of a segment. Also recall that $1$ is the first node of a short segment and $t$ is the first node of a long segment. Therefore, $[t\oplus (r+1)]$ is a long segment, and for all $r+1\le \ell' \le c-t+1$, node $t+\ell'$ is in a long segment if and only if $\ell'$ belongs to a long segment. Therefore, we conclude that
    \begin{align*}
        y([t\oplus (\ell+1)]) &= y([t\oplus (r+1)]) + y([t+r+1, t+\ell]) = \frac{1}{w} +  y([t+r+1, t+\ell])\\  &=  \frac{1}{w} +  y([r+1, \ell]) = y([1,\ell]).
    \end{align*}
    The lemma follows.
\end{proof}

With these properties we can bound the probability mass captured by an efficient strategy. The derived properties are analogous to the ones obtained in Lemmas~\ref{lm:yshift} and~\ref{lm:ymerge} for the not coprime case. 

\begin{lemma}\label{lm:yCtmax}
For all $v\in V\setminus\{1\}$ it holds that $y(C(T_v))\le h/w$.
\end{lemma}
\begin{proof}
    It suffices to prove that $y([v\oplus c]_{n-1})\le h/w$ for all $v$, as $y_{n-1}=0$ by construction. Consider first the case that $v \in \{1, \ldots, n-c-1\}$. By Lemma~\ref{lm:yProperties}, $y([v\oplus c]_{n-1})=y([1\oplus c]_{n-1})$ and node $c$ is the last node of the $h$-th segment. As all segments accumulate a mass of $1/w$, we obtain that $y([v\oplus c]_{n-1})=h/w$.  

    For the case $v\ge n-c-2$, then we observe that $y(C(T_v))= y([v,n-2]) + y([1,c+v-n])= y([v\oplus (n-v-1)]) + y([1\oplus (c+v-n)])$. Using Lemmas~\ref{lm:yExtrema} and~\ref{lm:ySpread}, we have that $y([v\oplus (n-v-1)])\le y([1\oplus (n-v-1)])\le y([t\oplus (n-v)])\le y([(c+v-n+1)\oplus(n-v)])$, and thus
    \begin{align*}
    y(C(T_v))& \le y([1\oplus (c+v-n)]) + y([(c+v-n+1)\oplus(n-v)]) \\
     & = y([1\oplus c])=h/w.
    \end{align*}    
\end{proof}

\begin{lemma}\label{lm:ymergeCoprime}
Consider two disjoint intervals $[u\oplus\ell]_{n-1}$ and $[v\oplus s]_{n-1}$ for $\ell+s\le c$. Then 
$$y([u\oplus \ell]_{n-1})+y([v\oplus s]_{n-1})\le y([u\oplus (\ell+s+1)]_{n-1}).$$
\end{lemma}
\begin{proof} We argue similarly as in the previous proof. Assume without loss of generality that if $0 \in [u\oplus \ell]_{n-1} \cup [v\oplus s]_{n-1}$ then $0 \in [u\oplus \ell] $, and thus $[v \oplus s]\subseteq [1,n-2].$ Then, using Lemmas~\ref{lm:yExtrema} and \ref{lm:ySpread},
\begin{align*}
y([u\oplus \ell]_{n-1})+y([v\oplus s]_{n-1}) & \le y([1\oplus s])+y([u\oplus \ell]_{n-1}) \\
 & \le y([t\oplus (s+1)])+y([u\oplus \ell]_{n-1}) \\
 & \le y([(u+\ell\mod n-1 )\oplus (s+1)])+y([u\oplus \ell]_{n-1}) \\
 &= y([u\oplus (\ell+s+1)]_{n-1}).
\end{align*}
\end{proof}

Now we have all tools needed to conclude our main theorem.

\begin{theorem} 
The value of the game is $\frac{h}{w}$. In particular Algorithm~\ref{alg:greedySeeker} yields an optimal solution for the seeker and \eqref{eq:yCoprime} and \eqref{eq:SegmentRule} yield an optimal solution for the hider for the coprime and not coprime case, respectively.
\end{theorem}
\begin{proof}
By Theorem~\ref{thm:not-cp} it suffices to argue about the coprime case. First, let us argue that the corresponding $y$ solution is indeed a probability distribution, that is, $y([0,n-1])=1$. Indeed, $y([0,n-1]) = y([1,n-2])$, and by Lemma~\ref{lm:yProperties} part~\ref{it:e}) we have that $y([1,n-2])=w \cdot \frac{1}{w}=1$ as each segment has a total mass of $1/w$ by construction.

With the same argument as in the proof of Theorem~\ref{thm:not-cp}, using Lemma~\ref{lm:ymergeCoprime}, for any $T\in \mathcal{T}_k$ with $C(T)$ maximal with 
$$C(T)=[u_1 \oplus \ell_1] \cup \ldots \cup [u_s \oplus \ell_s],$$
as in Proposition~\ref{prop:non-dominated}, we have that $y(C(T))\le y(C(T_{v_1}))$, where $v_1\neq 1$. By Lemma~\ref{lm:yCtmax}, $y(C(T))\le h/w$ for all efficient strategy $T\in \mathcal{T}_k$. We conclude that the objective value of $y$ is at most $h/w$. As our constructed solution $x$ for the primal also attains this objective, we conclude that both solutions must be optimal by weak duality.
\end{proof}

\section{A Dynamic Programming Algorithm for the General Case }\label{sec:DP}

In this section we consider problems [P] and [D] for the case when $G$ is a general tree and the profit $p$ is an arbitrary non-decreasing function.

A natural approach for computing the optimal profit is using the Ellipsoid method on the dual. For this, we need an algorithm to separate the first set of inequalities of [D], i.e., we need to solve $\max_{\T \in \mathcal{T}_k} \sum_{v \in V}y_{v} \cdot p(h_T(v))$ for a given $y\in \Delta_V$. In other words, we need to solve the \emph{best-response} problem for the seeker. 
This problem turns out to be NP-hard. 

\begin{theorem}\label{thm:bestResponseHardness}
Computing the best-response for the seeker is NP-hard, even if $G$ has constant diameter or constant degree, and if $p(t)=n-t$ for all $t\in \{0,\ldots,n\}$.
\end{theorem}
\begin{proof}
Consider the case where $k=n$ and $p(t)=n-t$ for all $t\in \{0,\ldots,n\}$. Observe that in this case the best-response problem is $\max_{\T \in \mathcal{T}_n} \sum_{v \in V}y_{v} \cdot h_T(v)$, that is, we aim to find $\T$ (of any height), that maximizes 
$$
\sum_{v \in V}y_{v} \cdot (n-h_T(v)) = n-\sum_{v \in V}y_{v} h_T(v).
$$
In other words, as $n$ is a constant, we aim to minimize 
the expected time (i.e., number of queries) needed to find the target $v^*$ if $\mathbb{P}(v^*=v)=y_v$ for all $v\in V$. This problem was shown to be NP-hard, even if $G$ has diameter at most 4 or if it has degree at most 16~\cite{cicalese_complexity_2011}.   
\end{proof}

In the remaining of this section we show that the best-response problem can be solved in time $O(n^22^{2k})$ for an arbitrary tree $G$. With the Ellipsoid algorithm~\cite{grotschel_ellipsoid_1981}, this implies an algorithm with time complexity $\text{poly}(n)2^{O(k)}$, which is fixed-parameter tractable on $k$.
In order to use the Ellipsoid method, we are interested in separating the first set of inequalities of linear program [D]. This poses the following problem: given a tree $G=(V, E)$, a mixed strategy of the hider $y \in \Delta_V$ and a non-increasing profit function $p: \{1, \ldots, k+1\} \to \N_0$ with $p(k+1)=0$, what is the search strategy in $\BSTs_k$ that maximizes the expected profit of the seeker? In this section, we introduce a dynamic program to solve this question, based on a characterization of search strategies using edge labelings.

\subsection{Edge Labelings to describe Search Strategies} \label{sec:edge-labelings}

Edge labelings are functions that map edges of a graph to numbers. In what follows, we specify the conditions an edge labeling must obey in order to properly encode a search strategy in $\BSTs_k$. Such a labeling will be called \emph{valid}. Subsequently, we formulate the objective function of the best-response problem in terms of valid labelings. Finally, we provide a characterization of valid labelings that will be used by the dynamic program defined in the next subsection.

The relationship between edge labelings and (unrestricted) search strategies was established independently in \cite{dereniowski2008edge, onak2006generalization}. We give a modified definition that captures the limited height of the search strategies in $\BSTs_k$.
\begin{definition}[Valid Labeling] \label{def:validlabel} A function $f: E \to \{0, \ldots, k\}$ is a \emph{valid labeling} for tree $G=(V,E)$ if for each pair of distinct edges $e_1, e_2 \in E$ such that $f(e_1) = f(e_2) > 0$, there is an edge $e_3$ on the simple path from $e_1$ to $e_2$ for which $f(e_3) > f(e_1)$. The set of valid labelings with range $\{0, \ldots, k\}$ is denoted by $\mathcal{F}_k$
\end{definition}
We remark that every valid labeling reaches its maximum at a unique edge trivially, except for the degenerate all zero labeling. 
Intuitively, a valid labeling maps an edge $e$ to the remaining budget of a search strategy right before it queries $e$ (see Figure \ref{fig:def-search-tree} for an example). Observe that a search strategy finds a vertex the moment all of its incident edges are queried. This suggests the following notation: for a vertex $v \in V$ and a valid labeling $f \in \mathcal{F}_k$, let $h_f(v)=k + 1 - \min \{f(e) : e \in \delta(v)\}$. Consequently, the \emph{expected profit} of a valid labeling is given by
\begin{equation}\label{eq:labeling objective}
    \sum_{v \in V} y_v \cdot p\left(h_f(v)\right).
\end{equation}
Then, we can show that the best-response problem can be solved by maximizing this equation.
\begin{proposition}\label{prop:labeling=BST}
    The maximum expected profit of a valid labeling is equal to the objective value of $y$, that is,
    $$
    \max_{f\in \mathcal{F}_k} \sum_{v \in V} y_v \cdot p(h_f(v)) = \max_{T \in \BSTs_k} \sum_{v \in V} y_v \cdot p(h_T(v)).
    $$
\end{proposition}

We leave the proof of this proposition for the appendix. With this equivalence established, we turn our attention to obtaining a ``local'' description of valid labelings. Such a description will give us a recipe for constructing valid labelings in an algorithmic fashion. Let $G=(V, E)$ be a rooted out-tree and let $f:E \to \{0, \ldots, k\}$ be an arbitrary labeling (not necessarily valid). In this setting, an edge $e$ that points from vertex $u$ to $v$ is denoted $e=(u, v)$. We say that an edge $e'$ is \emph{visible} from edge $e$ if on the directed path from $e$ ending in $e'$, there is no other edge $e''$ such that $f(e'') > f(e')$. In other words, the edges visible from $e$ are those that are not ``screened'' by greater values of $f$. The \emph{visibility sequence} of $e$, denoted $L(e)$, is the enumeration in ascending order of the labels of edges visible from $e$. We remark that the first label of $L(e)$ equals $f(e)$. We extend the definition of visibility sequence to vertices. The visibility sequence of a vertex $u$, $S(u)$, is the union of the visibility sequences of its outgoing edges (see Figure \ref{fig:visibility}). The following result gives us the conditions that we need to impose on visibility sets to obtain a valid labeling.

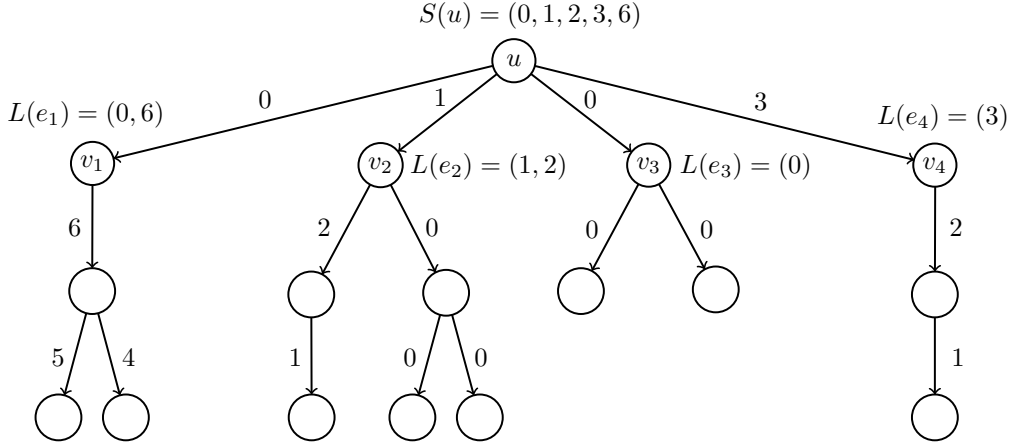
\begin{figure}[htb]
    \centering
    \tikzset{every picture/.style={line width=0.75pt}} %set default line width to 0.75pt        

\begin{tikzpicture}[x=0.75pt,y=0.75pt,yscale=-0.85,xscale=0.85]
%uncomment if require: \path (0,597); %set diagram left start at 0, and has height of 597

% Text Node
\draw    (358.22, 49.2) circle [x radius= 12.81, y radius= 12.81]   ;
\draw (358.22,49.2) node  [font=\small]  {$u$};
% Text Node
\draw    (279.13, 111.2) circle [x radius= 13.12, y radius= 13.12]   ;
\draw (279.13,111.2) node  [font=\small]  {$v_2$};
% Text Node
\draw    (108.22, 110.2) circle [x radius= 12.81, y radius= 12.81]   ;
\draw (108.22,110.2) node  [font=\small]  {$v_1$};
% Text Node
\draw    (438.22, 111.2) circle [x radius= 12.81, y radius= 12.81]   ;
\draw (438.22,111.2) node  [font=\small]  {$v_3$};
% Text Node
\draw    (608, 111.2) circle [x radius= 12.81, y radius= 12.81]   ;
\draw (608,111.2) node  [font=\small]  {$v_4$};
% Text Node
\draw    (108, 186) circle [x radius= 13.6, y radius= 13.6]   ;
\draw (102,178.4) node [anchor=north west][inner sep=0.75pt]    {};
% Text Node
\draw    (88, 261) circle [x radius= 13.6, y radius= 13.6]   ;
\draw (82,253.4) node [anchor=north west][inner sep=0.75pt]    {};
% Text Node
\draw    (128, 261) circle [x radius= 13.6, y radius= 13.6]   ;
\draw (122,253.4) node [anchor=north west][inner sep=0.75pt]    {};
% Text Node
\draw    (608, 188) circle [x radius= 13.6, y radius= 13.6]   ;
\draw (602,180.4) node [anchor=north west][inner sep=0.75pt]    {};
% Text Node
\draw    (608, 261) circle [x radius= 13.6, y radius= 13.6]   ;
\draw (602,253.4) node [anchor=north west][inner sep=0.75pt]    {};
% Text Node
\draw    (238, 188) circle [x radius= 13.6, y radius= 13.6]   ;
\draw (232,180.4) node [anchor=north west][inner sep=0.75pt]    {};
% Text Node
\draw    (318, 187) circle [x radius= 13.6, y radius= 13.6]   ;
\draw (312,179.4) node [anchor=north west][inner sep=0.75pt]    {};
% Text Node
\draw    (398, 186) circle [x radius= 13.6, y radius= 13.6]   ;
\draw (392,178.4) node [anchor=north west][inner sep=0.75pt]    {};
% Text Node
\draw    (478, 185) circle [x radius= 13.6, y radius= 13.6]   ;
\draw (472,177.4) node [anchor=north west][inner sep=0.75pt]    {};
% Text Node
\draw    (238.22, 261) circle [x radius= 13.6, y radius= 13.6]   ;
\draw (238.22,253.4) node [anchor=north] [inner sep=0.75pt]    {};
% Text Node
\draw    (298, 261) circle [x radius= 13.6, y radius= 13.6]   ;
\draw (292,253.4) node [anchor=north west][inner sep=0.75pt]    {};
% Text Node
\draw    (338, 261) circle [x radius= 13.6, y radius= 13.6]   ;
\draw (332,253.4) node [anchor=north west][inner sep=0.75pt]    {};
% Text Node
\draw (205,65.4) node [anchor=north west][inner sep=0.75pt]  [font=\small]  {$0$};
% Text Node
\draw (309,64.4) node [anchor=north west][inner sep=0.75pt]  [font=\small]  {$1$};
% Text Node
\draw (398,65.4) node [anchor=north west][inner sep=0.75pt]  [font=\small]  {$0$};
% Text Node
\draw (499,66.4) node [anchor=north west][inner sep=0.75pt]  [font=\small]  {$3$};
% Text Node
\draw (92,141.4) node [anchor=north west][inner sep=0.75pt]  [font=\small]  {$6$};
% Text Node
\draw (240,141.4) node [anchor=north west][inner sep=0.75pt]  [font=\small]  {$2$};
% Text Node
\draw (304,141.4) node [anchor=north west][inner sep=0.75pt]  [font=\small]  {$0$};
% Text Node
\draw (399,142.4) node [anchor=north west][inner sep=0.75pt]  [font=\small]  {$0$};
% Text Node
\draw (467,142.4) node [anchor=north west][inner sep=0.75pt]  [font=\small]  {$0$};
% Text Node
\draw (615,141.4) node [anchor=north west][inner sep=0.75pt]  [font=\small]  {$2$};
% Text Node
\draw (616,218.4) node [anchor=north west][inner sep=0.75pt]  [font=\small]  {$1$};
% Text Node
\draw (223,218.4) node [anchor=north west][inner sep=0.75pt]  [font=\small]  {$1$};
% Text Node
\draw (124,218.4) node [anchor=north west][inner sep=0.75pt]  [font=\small]  {$4$};
% Text Node
\draw (82,218.4) node [anchor=north west][inner sep=0.75pt]  [font=\small]  {$5$};
% Text Node
\draw (291,218.4) node [anchor=north west][inner sep=0.75pt]  [font=\small]  {$0$};
% Text Node
\draw (333,218.4) node [anchor=north west][inner sep=0.75pt]  [font=\small]  {$0$};
% Text Node
\draw (122,129.4) node [anchor=north west][inner sep=0.75pt]    {};
% Text Node
\draw (300,12.4) node [anchor=north west][inner sep=0.75pt]  [font=\small]  {$S( u) =( 0,1, 2, 3,6)$};
% Text Node
\draw (572,73.4) node [anchor=north west][inner sep=0.75pt]  [font=\small]  {$L(e_4) =( 3)$};
% Text Node
\draw (455,102.5) node [anchor=north west][inner sep=0.75pt]  [font=\small]  {$L(e_3) =( 0)$};
% Text Node
\draw (295,102.5) node [anchor=north west][inner sep=0.75pt]  [font=\small]  {$L(e_2) =( 1,2)$};
% Text Node
\draw (56,71.4) node [anchor=north west][inner sep=0.75pt]  [font=\small]  {$L(e_1) =(0,6)$};
% Connection
\draw[->]    (345.77,52.24) -- (120.66,107.16) ;
% Connection
\draw[->]    (348.14,57.1) -- (289.46,103.1) ;
% Connection
\draw[->]    (368.34,57.04) -- (428.1,103.35) ;
% Connection
\draw[->]    (370.65,52.28) -- (595.79,108.11) ;
% Connection
\draw[->]    (111.51,199.15) -- (124.49,247.85) ;
% Connection
\draw[->]    (104.49,199.15) -- (91.51,247.85) ;
% Connection
\draw[->]    (108.18,123) -- (108.04,172.4) ;
% Connection
\draw[->]    (608,201.6) -- (608,247.4) ;
% Connection
\draw[->]    (432.15,122.48) -- (404.44,174.02) ;
% Connection
\draw[->]    (444.3,122.47) -- (471.55,173.02) ;
% Connection
\draw[->]    (272.94,122.77) -- (244.42,176.01) ;
% Connection
\draw[->]    (285.12,122.88) -- (311.79,174.89) ;
% Connection
\draw[->]    (314.45,200.13) -- (301.55,247.87) ;
% Connection
\draw[->]    (321.55,200.13) -- (334.45,247.87) ;
% Connection
\draw[->]    (238.04,201.6) -- (238.18,247.4) ;
% Connection
\draw[->]    (608.18,124) -- (608.04,174.4) ;

\end{tikzpicture}
    \caption{An out-tree rooted at $u$ with an edge labeling. For $i=1, \ldots, 4$, edge $(u, v_i)$ is denoted by $e_i$.}
    \label{fig:visibility}
\end{figure}

\begin{proposition} \label{prop: visibility}
    \label{prop:equiv1}
A labeling $f:E \to \{0, \ldots, k\}$ of tree $G=(V, E)$ is valid if and only if for an arbitrary orientation of the tree it holds that
    \begin{enumerate}
        \item for every $u \in V$, if $e=(u', u) \in E$ then $f(e) \not \in S(u)$ or $f(e) = 0$ (or both), and
        \item for every $u \in V$, for distinct edges $e, e' \in \delta^+(u)$, the sets $L(e)$ and $L(e')$ are disjoint, except possibly for label $0$.
    \end{enumerate}
\end{proposition}

\begin{proof}
    Let $f$ be a valid labeling, and give an arbitrary orientation to $G$. Let $u \in V$ be a vertex and assume that there exists an edge $e=(u', u) \in E$. For property one, assume by contradiction that $l = f(e) \in S(u)$ and that $l > 0$. By the definition of visibility, there is an edge $e'$ at the end of a directed path $P$ that starts at $u$, such that $l = f(e') \geq f(e'')$ for all $e'' \in P$. However, this contradicts the assumption that $f$ is a valid labeling. Now, let us prove Property 2. Let $u \in V$ be a vertex and let $e_1, e_2 \in \delta^+(u)$ be outgoing edges of $u$. Again by contradiction, assume that there is a label $l>0$ that belongs to both $L(e_1)$ and $L(e_2)$. Let $e_1'$ and $e_2'$ be the edges visible from $e_1$ and $e_2$, respectively, such that $l = f(e_1') = f(e_2')$. Since $f$ is a valid labeling, there exists $e_3$ on the path that connects $e_1'$ and $e_2'$ such that $f(e_3)>l$. However, edge $e_3$ lies either in the directed path from $e_1$ to $e_1'$ or in the directed path from $e_2$ to $e_2'$. Assume that it lies in the former path without loss of generality. Then $e_1'$ is not visible from $e_1$, which contradicts our initial assumption. 

    Now assume that $f$ satisfies properties 1 and 2. Let $e_1, e_2 \in E$ be edges with the same positive label. If they are the endpoints of a directed path, Property 1 guarantees that there is an edge $e_3$ that connects $e_1$ and $e_2$ such that $f(e_1)$. In the case where there is no directed path that connects $e_1$ and $e_2$, then let $u$ be their common ancestor. In this case, Property 2 guarantees the existence of $e_3$ in the path that connects $e_1$ with $e_2$ such that $f(e_3) > f(e_1)$. We conclude that $f$ is a valid labeling.
\end{proof}

\subsection{Dynamic Program}
The idea of the dynamic program is to compute an optimal valid labeling from the bottom up, following an arbitrary orientation of the tree. For every edge and vertex, we need to keep track of their visibility sets and make sure that the conditions of Proposition~\ref{prop:equiv1} are satisfied. 

The dynamic program consists of two tables $B$ and $C$, with recurrences interleaving one another. For an edge $e=(u,v)$, consider the sub-tree consisting of this edge and the sub-tree rooted at $v$. By $B[e,L]$ we denote the maximum expected profit of a valid labeling of this sub-tree with visibility sequence $L$ for $e$. Vertex $u$ is not counted in the profit. Note that the label of $e$ is $\min L$.

Now, for some vertex $v$ with outgoing edges $e_1, \ldots, e_d$ and an integer $1 \leq i \leq d$, consider the subtree consisting of the edges $e_1,\ldots,e_i$ together with the sub-trees attached to the endpoints of these edges. By $C[v,i,S]$ we denote the maximum profit of a valid labeling of this sub-tree with visibility sequence $S$ at vertex $v$. Again vertex $v$ is not counted in the profit. For convenience we include the degenerate case $i=0$ corresponding to an empty sub-tree, which includes the case when $v$ is a leaf.

\begin{figure}[htb]
    \centering
    \begin{subfigure}[t]{0.45\textwidth}
        \centering
        \tikzset{every picture/.style={line width=0.75pt}} %set default line width to 0.75pt        

\begin{tikzpicture}[x=0.75pt,y=0.75pt,yscale=-0.7,xscale=0.7]
%uncomment if require: \path (0,585); %set diagram left start at 0, and has height of 585

%Flowchart: Extract [id:dp9930147882203486] 
\draw[gray]   (202,270) -- (252,350.5) -- (152,350.5) -- cycle ;
%Flowchart: Extract [id:dp6579209689222607] 
\draw[gray]   (499,270) -- (549,350.5) -- (449,350.5) -- cycle ;
%Flowchart: Extract [id:dp09937003953564427] 
\draw  [fill=black!50] (350,269.5) -- (400,350) -- (300,350) -- cycle ;

% Text Node

\node[draw, circle, gray] (c) at (350, 188.2) {$u$};
% Text Node
\node[draw, circle, gray, minimum height=0.6cm] (p) at (413.22,106.7)    {$ $};
% Text Node
\node (t2) at (350, 275) {};
% Text Node
\node (t1) at (196, 273)    {};
% Text Node
\node (t3) at (505, 274)   {};
% Text Node
\draw (360,235.2) node {$e$};
% Connection

\begin{scope}[every edge/.style={->, draw}]
    \path[gray, thin] (p) edge (c);
    \path[gray, thin] (c) edge (t1);
    \path[black, thick] (c) edge (t2);
    \path[gray, thin] (c) edge (t3);
\end{scope}

%\draw[->, thin, gray]    (404.22,118.25) -- (357.12,178.7) ;
% Connection
%\draw[->, thick]    (350,200.24) -- (350,268.5) ;
% Connection
%\draw[->, thin, gray]    (339.2,194.06) -- (204,269.48) ;
% Connection
%\draw[->, thin, gray]    (360.28,193.98) -- (497.5,269.07) ;

\draw (350, 373) node {$B[e, L]$};

\end{tikzpicture}
    \end{subfigure}
    \hspace{0.02\textwidth}
    \begin{subfigure}[t]{0.45\textwidth}
        \centering
        \tikzset{every picture/.style={line width=0.75pt}} %set default line width to 0.75pt        

\begin{tikzpicture}[x=0.75pt,y=0.75pt,yscale=-0.68,xscale=0.68]
%uncomment if require: \path (0,405); %set diagram left start at 0, and has height of 405

%Flowchart: Extract [id:dp6662389879439079] 
\draw[fill=black!50]   (160,310) -- (195,390) -- (125,390) -- cycle ;
%Flowchart: Extract [id:dp7239160010407644] 
\draw[fill=black!50]   (235,310) -- (270,390) -- (200,390) -- cycle ;
%Flowchart: Extract [id:dp11993178530442816] 
\draw[fill=black!50]   (335,310) -- (370,390) -- (300,390) -- cycle ;
%Flowchart: Extract [id:dp954108335835911] 
\draw[color=gray]   (435,310) -- (470,390) -- (400,390) -- cycle ;

% Text Node
%\draw[fill=black!50]    (326.12, 229) circle [x radius= 12.04, y radius= 12.04]   ;

\node[draw, circle, gray] (v) at (326.12,229) {$v$};
% Text Node

\node[draw, circle, gray, minimum height=0.6cm] (p) at (389.62,147.5) {};

\node (t1) at (153, 313)  {};
\node (t2) at (228, 317)  {};
\node (ti) at (336, 317)  {};
\node (td) at (443, 315)  {};

\draw (167.42,292) node  [font=\footnotesize]  {$e_{1}$};
\draw (240,292) node  [font=\footnotesize]  {$e_{2}$};
\draw (323,292) node  [font=\footnotesize]  {$e_{i}$};
\draw (285,350) node  [font=\smaller]  {$\cdots$};
\draw (385,350) node  [font=\smaller]  {$\cdots $};
\draw[gray] (394,292) node  [font=\footnotesize]  {$e_{d}$};

\begin{scope}[every edge/.style={->, draw}]
    \path[gray, thin] (p) edge (v);
    \path[black, thick] (v) edge (t1);
    \path[black, thick] (v) edge (t2);
    \path[black, thick] (v) edge (ti);
    \path[gray, thin] (v) edge (td);
\end{scope}

% Connection
%\draw[->, gray, thin]    (380.62,159.05) -- (333.52,219.5) ;
% Connection
%\draw[->, thick]    (329.1,240.66) -- (346.63,309.2) ;
% Connection
%\draw[->, thick]    (315.3,234.28) -- (160,310.11) ;
% Connection
%\draw[->, gray, thin]    (336.13,235.69) -- (447.1,309.98) ;
% Connection
%\draw[->, thick]    (317.72,237.62) -- (246.8,310.46) ;

\draw (270, 414.5) node {$C[v, i, S]$};

\end{tikzpicture}
    \end{subfigure}
    \caption{Illustration of table values $B$ and $C$. The shaded areas correspond to the portion of the tree that is considered by the table.}
    \label{fig:main}
\end{figure}
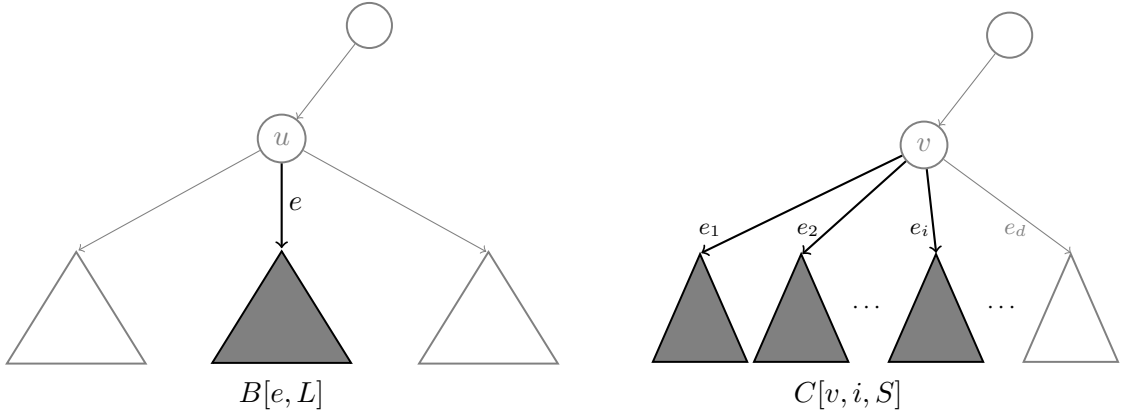

To introduce the recursion for $B$, let $e=(u, v)$ be an edge and let $L$ be a non-empty visibility sequence and denote $l_1=\min L$. Let $d$ be the degree of $v$, which is $0$ in case $v$ is a leaf. If $l_1=0$, vertex $v$ is not covered and we have the recursion
\begin{equation} \label{eq: table B0}
    B[e, L] = \max\{C[v,d,L], C[v,d,L\setminus(0)]\}.
\end{equation}
If $l_1>0$, we have the recursion
\begin{equation} \label{eq: table B1}    
B[e, L] = 
        \max_{S' \subseteq [l_1]} y(v) \cdot p(k+1-\min\{S' \cup (l_1)\}) + C[v, d, (S' \cup L) \backslash (l_1)],
\end{equation}
where $[l_1]$ denotes the set $\{0,1,\ldots,l_1-1\}$.  For the second case, $S'$ corresponds to the set of labels screened by $l_1$, but present in the visibility sequences of outgoing edges from $v$. The label $l_1$ is positive, so it must be excluded from visibility sequence of $v$ to satisfy Property 1 of Proposition \ref{prop:equiv1}. To prevent $e$ to be unlabeled, we define $B[e, \emptyset] = -\infty$. Next, we focus on table $C$. We define

For the special case $i=0$ of table $C$ we define 
\begin{equation}  \label{eq: table C0}
C[v,0,S]=
\begin{cases}
0 &\text{if } S=\emptyset \\
-\infty & \text{if }S\neq \emptyset.
\end{cases}
\end{equation}
For an index $1\leq i\leq d$ we define
\begin{equation} \label{eq: table C1}
    C[v, i, S] = 
        \max_{L_i \subseteq S} B[e_i, L_i] + 
        \max\{ C[v, i-1, S \backslash L_i ], C[v, i-1, S \backslash (L_i \backslash (0))]\}.
\end{equation}
The idea is to decompose the visibility sequence $S$ of vertex $v$ in almost disjoint visibility sequences $L_i$ for each edge $e_i$, with the potential exception of label $0$. This is done in order to preserve the second property of Proposition \ref{prop: visibility}. We achieve that by assigning $L_i \subseteq S$ to edge $e_i$ and calling table $C$ again on the remaining edges with updated visibility sequence $S \backslash (L_i \backslash (0))$. The maximum of two choices translates the fact that $e_i$ could be the only edge among $e_1,\ldots,e_i$ to be labeled 0, or not the only one. In both cases $L_i$ would contain $0$, but in the second case determines whether the recursion in $C[v,i-1,\cdot]$ continues also with a sequence starting with 0 or not.

The maximum value of Equation \eqref{eq:labeling objective} over all valid labelings can be computed by solving
\begin{equation} \label{eq: DP equation}
\max_{S} y(r) \cdot p(k+1-\min S) +  C[r, d, S],    
\end{equation}
where $d$ is the degree of the fixed root $r$. Notice, however, that the resulting labeling might not be maximal. This is easily solved by making the algorithm break ties in favor of lexicographically larger visibility sequences. This, together with the fact that $p$ is non-increasing, results in a maximal valid labeling. By following the dynamic program constructed through equations \ref{eq: table B0}, \ref{eq: table B1} and \ref{eq: table C0}, we obtain the following theorem :

\begin{theorem}
    The best-response problem can be solved in $\mathcal{O}(n^2 2^{2k})$.
\end{theorem}
\begin{proof}
For the given dynamic program, there are $O(n 2^k)$ variables, as both tables $B$ and $C$ are described by a single edge or vertex and a label set. Each can be computed in time $O(2^k)$, in bottom up order of the tree. Therefore, Equation \ref{eq: DP equation} can be solved in time $\mathcal{O}(n^2 2^{2k})$. 
\end{proof}

\appendix

\section{Missing Proofs of Section~\ref{sec:line}}

\subsection{Proof of Proposition~\ref{prop:non-dominated}}

The following lemmas will be useful to show Proposition~\ref{prop:non-dominated}.

\begin{lemma}\label{lm:leaves}
    Let $G=(V, E)$ be a path graph with $n$ vertices. We say that $p_0, p_1, \ldots, p_{j} \in \N_0$ with $0 = p_0  < p_1 < \cdots < p_{j-1} < p_{j} = n$ define a \emph{partition into $j$ intervals} of $V$. Then, $\{V(\ell): \lambda \in L\} = \{[p_i, p_{i+1}-1]: 0 \leq i \leq j-1\}$ for some search strategy $T \in \BSTs_k$ with leaves $L$ if and only if $j \leq 2^k$.
\end{lemma}
\begin{proof}
    If $T \in \BSTs_k$ is a search strategy with leaves $L$, then $\{V(\lambda) : \lambda \in L\}$ is trivially a partition into $|L|$ intervals, since each $V(\lambda)$ corresponds to a connected component of $V$, and connected components of path graphs are intervals. Since $T$ has height at most $k$, then $|L| \leq 2^k$. 
    
    For the other direction, we prove by strong induction on $k$ that for every path graph $G=(V, E)$ with $n$ vertices and for every partition into $j$ intervals $p_0, \ldots, p_j$, if $j \leq 2^k$ then there exists $T \in \BSTs_k(G)$ with leaves $L$ such that $\{V(\lambda): \lambda \in L\} = \{[p_i, p_{i+1}-1]: 0 \leq i \leq j\}$, where $\BSTs_k(G)$ is the set of search strategies with height at most $k$ for $G$. For $k=0$, the result is trivial. Assume now that the desired result holds for all $k' < k$, with $k \geq 2$ being some natural number. Let $G=(V, E)$ be a path graph with $n \geq 2^k$ and let $p_0, \ldots, p_j$ be a partition into $j$ intervals, with $j \leq 2^k$. Without loss of generality, assume $j > 2^{k-1}$, and select $i = 2^{k-1}$. Notice that $p_0, \ldots, p_{i}$ and ${p_{i}, p_{i+1}, \ldots, p_j}$ define partitions of at most $2^{k-1}$ intervals of $V(G_{p_i})$ and $V(G_{u_i + 1})$, respectively. From the inductive hypothesis, we obtain search strategies $T_1 \in \BSTs_{k-1}(G_{p_i})$ and $T_2 \in \BSTs_k(G_{p_i+1})$ with leaves $L_1$ and $L_2$, respectively, such that $\{V(\lambda): \lambda \in L_1\} = \{[p_l, p_{l+1}-1]: 0 \leq l \leq i\}$ and $\{V(\lambda): \lambda \in L_2\} = \{[p_l, p_{l+1}-1]: i \leq l \leq j\}$. Lastly, we define $T = (N, A)$ with root $\rho$ labeled $e(\rho) = \{p_i, p_i+1\}$. The two children of $\rho$ are defined as $T_1$ and $T_2$, so by definition $T \in \BSTs_k$. Also, the set of leaves of $T$ is $L_1 \cup L_2$, proving the lemma.
\end{proof}

As the proof of the lemma suggests, the partition into intervals of a given search strategy is simply a representation of the edge queries performed by the strategy on the graph in a compact format. 

\begin{lemma}\label{lm:non-dominated}
    Let $G=(V, E)$ be a path graph with $n$ vertices and let $T \in \BSTs_k$ be some search strategy, with $2^k < n$. For a vertex $v \in V$, denote the unique leaf $\lambda$ of $T$ such that $V(\lambda) \ni v$ as $\lambda(v)$. Then, $C(T)$ is a maximal covered set if and only if 
    \begin{enumerate}[a)]
        \item $T$ has exactly $2^k$ leaves, and
        \item for any vertex $v \in V$ such that $v-1, v \not \in C(T)$, $\lambda(v-1) = \lambda(v)$.
    \end{enumerate}
\end{lemma}
\begin{proof}
    For the first implication, we proceed by contrapositive. Let $T \in \BSTs_k$ be a search strategy and assume it has less than $2^k$ leaves. Recall that $C(T)$ are the vertices $v$ such that $V(\lambda(v)) = \{v\}$. By Lemma \ref{lm:leaves}, we can find search strategy $T' \in \BSTs$ with $2^k$ leaves such that it defines a partition with the same intervals and at least one more of length 1, hence $C(T') \supset C(T)$. Now assume that there exists a vertex $v \in V$ such that neither $v$ nor $v - 1$ belong to $C(T)$, but $\lambda(v-1) \neq \lambda(v)$. For $v' = \min\{\lambda(v-1)\}$, notice that $V(\lambda(v-1)) = [v', v-1]$, so $v'$ is not covered by $T$ either. We will find a search strategy that covers all the same vertices than $T$ plus vertex $v'$. Let $p_0, \ldots, p_j$ be the partition into $j$ intervals defined by $T$, then there is $1 \leq i \leq j-1$ such that $p_i = v$. Define a partition into $j$ intervals $q_0, \ldots, q_j$ as follows: 
    $$q_l = \begin{cases}
        p_l & \text{ if $l \neq i$},\\
        v'+1 & \text{ if $l = i$}.
    \end{cases}$$
    By Lemma \ref{lm:leaves}, there is a search strategy $T' \in \BSTs$ that induces this partition. As the definition of the partition suggests, there is a leaf $\lambda'$ of $T'$ such that $V(\lambda') = \{v'\}$, namely the leaf associated to interval $[q_{i-1}, q_i-1]$. Moreover, all other intervals of length 1 are preserved, so if $v \in C(T)$, then $v\in C(T')$ as well. We conclude that $C(T') \supset C(T)$.

    For the other direction, let $T \in \BSTs_k$ such that Properties a) and b) hold. Let $p_0, \ldots, p_{2^k}$ be the partition into intervals associated to $T$. We remark that Property b) implies that for all $i \in \{1, \ldots, 2^k-1\}$, we have that either $p_i = p_{i-1} + 1$ or $p_i = p_{i+1} - 1$. We also remark that $v \in C(T)$ if and only if there exists $i \in \{0, \ldots, 2^k-1\}$ such that $u_i = v$ and $u_i = u_{i+1}-1$. Now, assume by contradiction that there is a search strategy $T' \in \BSTs_k$ such that $C(T') = C(T) \cup \{v\}$, for some $v \not \in C(T)$. Let $q_0, \ldots, q_j$ be the partition into intervals induced by $T'$. For the two remarks mentioned initially, we have that for every $1 \leq i \leq 2^k-1$ there has to be some $1 \leq l \leq j-1$ such that $p_i = q_l$. Otherwise, there would be a vertex that is covered only by $T$. This implies that $j \geq 2^k$. Now, because $v \not \in C(T)$, assume without loss of generality that there is no $1 \leq i \leq 2^k - 1$ such that $p_i = v$. We know that there exists $1 \leq l \leq j-1$ with $q_l = v$, because $v \in C(T')$. This implies that $j > 2^k$, which contradicts that $T' \in \BSTs_k$. 
\end{proof}

Now we have all the tools to prove Proposition \ref{prop:non-dominated}.
\begin{proof}[Proof of Proposition \ref{prop:non-dominated}]
    Let $C(T)$ be a maximal covered set for some $T \in \BSTs_k$, and let $[u_1 \oplus \ell_1], \ldots , [u_s \oplus \ell_s]$ be maximal intervals such that their union equals $C(T)$. A trivial corollary of Lemma \ref{lm:leaves} is that if a vertex $v$ is not covered by $T$, then there is at least one neighbour of $v$ that is also not covered by $T$. Using the contrapositive of this fact on vertex $v=0$, we deduce that if $1 \in C(T)$, then $0 \in C(T)$ as well. Thereby no maximal interval can start at vertex 1. The same reasoning is used to discard maximal intervals ending at vertex $n-2$, hence proving Property i). Next, observe that Property ii) is equivalent to the mentioned corollary restricted to all vertices other than $0$ and $n-1$, which as we mentioned follows directly from Lemma \ref{lm:leaves}. To show Property iii), we go by cases. If $\{0, n-1\} \cap C(T) = \emptyset$, then we know of at least two leaves of $T$ associated to intervals of length at least 2: one containing vertex 0 and another one containing vertex $n-1$. By Property b) of Lemma \ref{lm:non-dominated}, there is also exactly one leaf associated to leafs of length at least two in between every maximal interval, that is, exactly $s-1$ such leaves. As a consequence, we have that
    $$
        \sum_{t=1}^s \ell_s = |C(T)| = 2^k - 2 - (s - 1) = c+1-s,
    $$
    where we use that $T$ has exactly $2^k$ leaves by Property a) of Lemma \ref{lm:non-dominated}. In the case where $\{0, n-1\} \cap C(T) \neq \emptyset$, then we have exactly one leave of length at least  alongside each maximal interval. Accordingly, we have
    $$
        \sum_{t=1}^s \ell_s = |C(T)| = 2^k - s = c+1-s.
    $$

    Now we address the other direction. We will construct a partition into intervals using the description of $C$ and we will show that the tree $T$ defined is such that $C(T) = C$ and that it has Properties a) and b) of Lemma \ref{lm:non-dominated}. Start by defining $p_0 = 0$. If $\{0, n-1\} \cap C = \emptyset$, then for each maximal interval $1 \leq t \leq s$ construct $p_i^t = u_t + i$ for $0 \leq i \leq \ell_t$. In the case where $\{0, n-1\} \cap C \neq \emptyset$, then for each $1 \leq t \leq s$ construct $p_i^t = (u_t + i) \mod{n}$ for $0 \leq i \leq \ell_t$, unless the expression equals 0, which happens exactly once by the fact that the intervals are maximal modulo $n$. When this is the case, simply omit the construction of $p_i^t$. Afterwards, sort the constructed $p_i^t$ in a unique sequence of numbers $p_1 < \ldots < p_{j-1}$ and complete the partition into intervals by defining $p_j = n$. We can guarantee that every constructed number is distinct from each other because the intervals that define $C$ are maximal. If $\{0, n-1\} \cap C = \emptyset$, we created $\ell_t + 1$ numbers for each $1 \leq t \leq s$. Consequently, by Property iii) we have
    \begin{equation*}
        j-1 = \sum_{t=1}^s (\ell_t + 1) = (c + 1 - s) + s = 2^k - 1, 
    \end{equation*}
    and if $\{0, n-1\} \cap C \neq \emptyset$, the same reasoning with the consideration of the omitted number gives
    \begin{equation}
        j-1 = \sum_{t=1}^s (\ell_t + 1) - 1 = (c + 2 - s) + s - 1 = 2^k - 1. 
    \end{equation}
    Either way, we use Lemma \ref{lm:leaves} to infer the existence of $T \in \BSTs_k$ that induces the defined partition into intervals and that satisfies Property a) of Lemma \ref{lm:non-dominated}. Moreover, for every $v \in C$, we constructed some $p_i$ where $v = p_i = p_{i+1} - 1$, i.e. $\{v\} = [p_i, p_{i+1}-1]$, therefore $v \in C(T)$ and we obtain that $C(T) \supseteq C$. Now, let $v$ be a vertex that is not contained in $C$. By Property i) and ii), we can assume without loss of generality that $v-1 \not \in C$ as well. By construction, there exists $i$ for which $p_i \leq v-1$ and $v \leq p_{i+1}-1$, that is, $v-1$ and $v$ are both contained in $[p_i, p_{i+1}-1]$. Then, we have that neither $v$ nor $v-1$ belong to $C(T)$, and $\lambda(v-1) = \lambda(v)$. This last equality implies Property b) of Lemma \ref{lm:non-dominated}. We conclude that $C(T)= C$ and that $C(T)$ is a maximal covered set.
    \end{proof}

 \subsection{Proof of Corollary~\ref{cor:efficientAlg}}
    
    As stated in Lemma~\ref{lm:bezout}, the values $h$ and $w$ can be computed in time $O(\log n)$ by using the Extended Euclidean algorithm. Take a number $t$ uniformly at random in $\{0,\ldots,w-1\}$ (using $O(\log w)\le O(\log n)$ bits), which represents the $(t+1)$-th tree inserted to $\X$ in the algorithm. If $t=0$ this tree corresponds to $\T_0$, otherwise, it corresponds to the efficient strategy $\T_v$ with $v= (t\cdot c \mod n-1) +1$. For a given $v$, the queries implied by $\T_v$ can be easily determined as in the proof of Proposition~\ref{prop:non-dominated}.

\section{Proof of Proposition ~\ref{prop:labeling=BST}}
    We first provide an auxiliary lemma.

\begin{lemma} \label{lm:labeling<=BST}
    Let $f \in \mathcal{F}_k$ be a valid labeling. There is a search strategy $T \in \BSTs_k$ such that for every edge $v \in V$, we have $h_T(v) \geq h_f(v)$.
\end{lemma}
\begin{proof}
    By induction on $k$, we prove that for every tree $G=(V, E)$ and for every valid labeling $f \in \mathcal{F}_k(G)$, there is a search strategy $T \in \BSTs_k(G)$, where $\mathcal{F}_k(G)$ is the set of all valid labelings with range $\{0, \ldots, k\}$  and $\BSTs_k(G)$ is the set of all search strategies with height at most $k$. Let $G=(V, E)$ be an arbitrary tree. For $k=1$, let $f \in \mathcal{F}_1(G)$ be a nontrivial valid labeling. Then, there is a unique edge $e^* \in E$ for which $f(e^*)=1$. Define a search strategy $T =(N, A) \in \BSTs_1(G)$ with root $\rho$ labeled $e(\rho) = e^*$. The only case where $h_f(v) = 1$ for some $v \in V$ is when $\delta(v) = \{e^*\}$. In this situation, there is a leaf $\lambda$ of $T$ such that $V(\lambda) = \{v\}$. Since the height of $T$ is 1, we have $h_T(v) = 1 = f_T(v)$, hence proving the base case. Now, assume that the desired result holds for $k-1$, with $k \geq 3$. Let $f \in \mathcal{F}_k(G)$ be a valid labeling and let $ e^* = uv = \argmax_{e \in E}f(e)$. Let $f|_{E(G_u)}$ be the labeling defined by $f$ restricted to edges in $E(G_u)$, and define $f|_{E(G_v)}$ analogously. It is easy to see that $f|_{E(G_u)} \in \mathcal{F}_{k-1}(G_u)$ and $f|_{E(G_v)} \in \mathcal{F}_{k-1}(G_v)$. By the inductive hypothesis, there are search strategies $T_u \in \BSTs_{k-1}(G_u)$ and $T_v \in \BSTs_{k-1}(G_v)$ such that the lemma holds. Then, define $T \in \BSTs_k(G)$ as a tree rooted in node $\rho$ labeled with edge $e(\rho)=e^*$, and let the children of $\rho$ be the subtrees $T_u$ and $T_v$. The result follows by construction.
\end{proof}
Notice that the proof of the lemma shows how to recover a search strategy from a valid labeling. Now we continue with the proof of the proposition.

\begin{proof}[Proof of Proposition \ref{prop:labeling=BST}]
    Let $T=(N, A) \in \BSTs$ be some search strategy. This strategy induces a labeling $f_T$ in the following way:
    $$
        f_T(e) = \begin{cases}
            k - h_T(\nu) & \text{for $\nu \in N$ s.t. $e(\nu) = e$}, \\
            0 & \text{else.}
        \end{cases}, \quad \text{for all } e \in E,
    $$
    where $h_T(\nu)$ gives the length of the path from the root of $T$ to $\nu$.
    We show that $f_T$ belongs to $\mathcal{F}_k$. Let $e_1, e_2 \in E$ be distinct edges such that $f_T(e_1) = f_T(e_2)>0$. By definition, there are (distinct) nodes $\nu_1, \nu_2 \in N$ such that $k - h_T(\nu_1) = k - h_T(\nu_2) = f(e_1)$. Since both nodes have the same height, they must have a common ancestor $\nu_3 \in N$. By definition, $\nu_3$ must be labeled with an edge $e_3$ such that $e_1$ and $e_2$ are in different connected components of $G - e_3$, which implies that $e_3$ lies in the path that connects $e_1$ with $e_2$. Moreover, $f_T(e_3) = k - h_T(\nu_3) > k - h_T(\nu_1) = f_T(E_1)$ and $f_T(e) \leq k$ for all $e \in E$, which proves that $f_T \in \mathcal{F}_k$. Now, it is easy to see that $h_T(v) = h_{f_T}(v)$ for all $v \in V$, from which we obtain
    $$
    \max_{f\in \mathcal{F}_k} \sum_{v \in V} y_v \cdot p(h_f(v)) \geq \max_{T \in \BSTs_k} \sum_{v \in V} y_v \cdot p(h_T(v)).
    $$

    To prove the other side, let $f \in \mathcal{F}_k$. We use Lemma \ref{lm:labeling<=BST} to obtain a search strategy $T \in \BSTs_k$ such that $h_T(v) \geq h_f(v)$ for all $v \in V$. It follows that
    $$
    \max_{f \in \mathcal{F}_k} \sum_{v \in V} y_v \cdot p(h_f(v)) \leq \max_{T \in \BSTs_k} \sum_{v \in V} y_v \cdot p(h_T(v)).
    $$
\end{proof}

\bibliography{bib}

\end{document}